\documentclass[acmsmall,nonacm]{acmart}

\usepackage{algorithm}
\usepackage{algorithmic}

\usepackage{amsmath}
\usepackage{amsthm}
\usepackage{amsfonts}
\usepackage{booktabs}
\usepackage{multirow}
\usepackage{xspace}

\usepackage{adjustbox}

\usepackage{todonotes}

\usepackage{cancel}
\usepackage{subcaption}
\usepackage{multirow}
\usepackage{rotating}
\usepackage{cleveref}
\usepackage{forest}
\forestset{
    default preamble={
        for tree={
            l sep=5pt,
            s sep=5pt,
            for children={l sep-=1.0em,l-=1.5em}
        }
    },
    hole/.style={draw=red, dashed},
    uhole/.style={draw=orange}
}

\newcommand{\G}{\mathcal{G}} 
\newcommand{\Int}{\text{Int}} 
\newcommand{\Bool}{\text{Bool}} 

\newcommand{\programs}{\mathbf{P}} 
\newcommand{\templatetrees}{\mathcal{T}} 

\theoremstyle{definition}
\newtheorem{definition}{Definition}
\newtheorem{theorem}{Theorem}
\newtheorem{lemma}{Lemma}

\usepackage{listings}

\definecolor{codegreen}{rgb}{0,0.6,0}
\definecolor{codegray}{rgb}{0.5,0.5,0.5}
\definecolor{codepurple}{rgb}{0.58,0,0.82}
\definecolor{backcolour}{rgb}{0.95,0.95,0.92}
\definecolor{varnodecyan}{rgb}{0.2784313725490196, 0.9529411764705882, 1.0}
\lstdefinestyle{mystyle}{
    backgroundcolor=\color{backcolour},   
    commentstyle=\color{codegreen},
    keywordstyle=\color{magenta},
    numberstyle=\tiny\color{codegray},
    stringstyle=\color{codepurple},
    basicstyle=\ttfamily\footnotesize,
    breakatwhitespace=false,         
    breaklines=true,                 
    captionpos=b,                    
    keepspaces=true,                 
    numbers=left,                    
    numbersep=5pt,                  
    showspaces=false,                
    showstringspaces=false,
    showtabs=false,                  
    tabsize=2
}

\lstdefinelanguage{Julia}{
    morekeywords={abstract,begin,break,case,catch,const,continue,do,else,elseif,end,export,false,for,function,global,if,import,importall,immutable,local,macro,module,otherwise,quote,return,struct,switch,true,try,type,typealias,using,while,
    },
    sensitive=true,
    morecomment=[l]{\#},
    morecomment=[n]{\#=}{=\#},
    morestring=[s]{"}{"},
    morestring=[m]{'}{'},
}

\setcopyright{acmlicensed}
\copyrightyear{2025}
\acmYear{2025}
\acmDOI{XXXXXXX.XXXXXXX}

\acmConference[Conference acronym 'XX]{Make sure to enter the correct
  conference title from your rights confirmation emai}{June 03--05,
  2018}{Woodstock, NY}
\acmISBN{978-1-4503-XXXX-X/18/06}

\newcommand{\solvername}{\textsf{BART}\xspace}
\newcommand{\cstrname}[1]{\texttt{#1}}

\begin{document}

\title{Modelling Program Spaces in Program Synthesis with Constraints}

\author{Tilman Hinnerichs}
\email{t.r.hinnerichs@tudelft.nl}
\authornote{Both authors contributed equally to this research.}
\orcid{1234-5678-9012}
\affiliation{%
  \institution{TU Delft}
  \city{Delft}
  \country{Netherlands}
}

\author{Bart Swinkels}
\email{b.j.a.swinkels@student.tudelft.nl}
\authornotemark[1]
\affiliation{%
  \institution{TU Delft}
  \city{Delft}
  \country{Netherlands}
}

\author{Jaap de Jong}
\email{j.dejong-18@student.tudelft.nl}
\affiliation{%
  \institution{TU Delft}
  \city{Delft}
  \country{Netherlands}
}

\author{Reuben Gardos Reid}
\email{r.j.gardosreid@tudelft.nl}
\affiliation{%
  \institution{TU Delft}
  \city{Delft}
  \country{Netherlands}
}

\author{Tudor Magirescu}
\email{magirescu@student.tudelft.nl}
\affiliation{%
  \institution{TU Delft}
  \city{Delft}
  \country{Netherlands}
}

\author{Neil Yorke-Smith}
\email{n.yorke-smith@tudelft.nl}
\affiliation{%
  \institution{TU Delft}
  \city{Delft}
  \country{Netherlands}
}
\author{Sebastijan Dumancic}
\email{s.dumancic@tudelft.nl}
\affiliation{%
  \institution{TU Delft}
  \city{Delft}
  \country{Netherlands}
}

\renewcommand{\shortauthors}{Hinnerichs et al.}

\begin{abstract}
A core challenge in program synthesis is taming the large space of possible programs.
Since program synthesis is essentially a combinatorial search, the community has sought to leverage powerful combinatorial constraint solvers.
Here, constraints are used to express the program semantics, but not as a potentially potent tool to remove unwanted programs.
Recent inductive logic programming approaches introduce constraints on the program’s syntax to be synthesized.
These syntactic constraints allow for checking and propagating a constraint without executing the program, and thus for arbitrary operators.
In this work, we leverage syntactic constraints to model program spaces, defining not just solutions that are feasible, but also ones that are likely useful.
To demonstrate this idea, we introduce \solvername, a solver that efficiently propagates and solves these constraints.
We evaluate \solvername on program space enumeration tasks, finding that, on the one hand, the constraints eliminate up to 99\% of program space, and that, on the other hand, investing in modelling program spaces pays off, reducing enumeration time significantly.
\end{abstract}

\ccsdesc[500]{Software and its engineering~Programming by example}
\keywords{Program Synthesis, Constraints, Automated Programming}

\received{20 February 2007}
\received[revised]{12 March 2009}
\received[accepted]{5 June 2009}

\maketitle


\section{Introduction}
\label{sec:intro}

The goal of program synthesis is simple: Let computers write their own code.
Given the user's intent and a target language, a program synthesis algorithm returns the target program.
Here, the intent is formalized by a \textit{specification} and language described by a \textit{grammar}.
Program synthesis is often framed as a combinatorial search: To find a solution, the synthesizer has to enumerate (possibly all) programs that follow the grammar.
Unfortunately, even with a relatively simple grammar, the spanned space of possible programs usually grows exponentially~\cite{GulwaniPS17}.
Hence, a core challenge in program synthesis is taming and restricting the large space of possible programs.
Even modelling the program space is a challenge:
On one hand, we want the program space to be expressive enough to tackle real-world problems, but on the other, also small enough so it can be searched effectively.

Combinatorial constraint solvers are a natural and common choice to approach program synthesis: 
They are efficient at combinatorial search, and constraints can easily express specifications and help to restrict program spaces by, e.g., removing redundancies. 
Thus, the community has sought to leverage powerful 
constraint solving paradigms, such as Satisfiability (SAT), Satisfiability Modulo Theory (SMT), and Answer Set Programming (ASP) solvers~\citep{SyGuSAlur13,cvc5,z3, Brahma2011Gulwani}.

However, the expressivity and current usage of constraints in program synthesis are limited.
In Syntax-Guided Synthesis (SyGuS)~\cite{SyGuSAlur13}, a common program synthesis paradigm, constraints express the program semantics and thus allow the solvers to execute programs.
This, however, requires knowing every operator’s semantics. 
Adapting to a new domain -- by adding a new operator or a new type of constraint, for instance -- thus depends on defining an entire theory and the operator's behaviour.
For example, defining the semantics of every operator of the programming language Julia into SMT is infeasible, due to loops, meta-programming, and others, which have no equivalent in SMT.
One can argue that this non-trivial step, namely encoding an entire programming language's semantics into constraints, prevented the community from moving beyond standard benchmarks~\cite{semgus} such as Linear Integer Arithmetic and String Theory~\cite{SMTlibBarrettMRST10}. 
Yet, it is intuitive to formulate \textit{useful} constraints when synthesizing loops in Julia:
For example, the expression after \texttt{while true} should contain a \texttt{break} statement.
While we will not go beyond common synthesis benchmarks in this paper, we aim to provide a framework that allows us to formulate such constraints.

A second use of constraints in program synthesis is as a potentially potent tool to remove unwanted programs, also called `pruning' in constraint-solving.  
In other fields, such as constraint satisfaction problem (CSP) solving~\cite{HandbookOFCP} or inductive logic programming (ILP) \cite{ILPat302022Cropper}, constraints are used to induce a \textit{language bias}~\cite{MetaruleReductionCropperT20,ProgolMuggleton95}. 
That is, to guide the search -- defining not just solutions that are feasible, but also ones that are \textit{likely useful}.
Specifically, in CVC~\cite{cvc5} and other SyGuS synthesizers, constraints are not used beyond expressing semantics and specification of a program synthesis problem. 
We aim to extend the benefits of this idea in program synthesis.

\paragraph{Syntactic constraints}
ILP, i.e., program synthesis for logic programs, has considered \textit{syntactic constraints}, a new type of constraint.
Recent ILP approaches, like Popper~\cite{POPPER}, introduce constraints on the program's syntax to be synthesized~\cite{MuggletonLT15,ILPat302022Cropper}.
Popper uses constraint solvers (namely ASP) to represent and enumerate program space, but does not leverage constraints to \textit{describe a program's semantics}.
For logic programs, the notion of constraints can be checked on a syntactical level, without executing the program.
We aim to generalize syntactic constraints for general program synthesis.

Using syntactic constraints has several advantages.
First, it allows us to treat operators as black boxes and formulate constraints over them; adding or changing operators becomes straightforward. 
ILP and Popper thus allow us to tackle a much wider variety of domains.
A second advantage of using syntactic constraints is: not having to evaluate programs to test constraints accelerates propagation -- a core of constraint-solving algorithms.
A third advantage is that many such constraints, such as breaking symmetries or forbidding a certain program structure, directly follow from the domain, are straightforward to derive, and are valid across problems. 

Generalizing syntactic constraints beyond logic programs, and how to best use solvers to express and iterate programs that do not violate them, is not obvious.
Firstly, common SyGuS solvers like CVC5 are neither suitable to represent nor to efficiently propagate syntactic constraints.
These solvers are tailored to efficiently represent semantically-related programs, e.g., evaluating to the same output.
Contrary, a syntactical constraint prevents programs with similar syntax from being generated.
For example, preventing the redundant addition of $0$ removes the syntactically overlapping programs $x+0$, $1+0$, and $(x+1)+0$.
To efficiently propagate these syntactic constraints, a solver requires an efficient representation of syntactically similar programs, which common SyGuS solvers do not provide.
Further, expressing and representing syntactic constraints in existing SyGuS solvers is hard.
For example, CVC5 is built around decision variables, which dictate the choice of constraints, but does not provide explicit access to them.
Thus, formulating purely syntactic constraints is challenging, if not impossible.

Secondly, approaches like Popper, too, fall short when applied to program synthesis directly.
First, Popper is limited to ILP and hence only synthesizes logic programs. 
Moreover, constraints in Popper and other synthesizers \cite{MuggletonLT15,PSusingCSPAhlgrenY13} are encoded at the propositional level, each removing a single or only a few programs.
Indeed, this is feasible for ILP, but not for general program synthesis, where many more constraints must be derived and each propagated individually.
We observe, however, that many discovered constraints are similar to each other in form and operators used.
If they were expressed on a first-order level, constraints could be represented more compactly and propagated more efficiently. 

\paragraph{This paper.}
This paper uses syntactic constraints to \textit{model} program spaces.
We aim to invest in defining the syntactic space of programs, beyond formulating a context-free grammar.
That is, we use syntactic constraints to prune the search space before (and during) the search.
We follow a central claim of ILP approaches: Many useful and highly impactful constraints can be expressed purely syntactically, also in general program synthesis.
To express these constraints, this paper contributes an extensible language of constraints that allows us to express syntactic constraints.

In line with this goal, this paper contributes to the literature a new constraint solver, named \solvername, tailored towards efficiently solving syntactical constraints in program synthesis.
Crucially, \solvername leverages the inference strength of syntactic constraints to remove invalid programs \textit{before they are enumerated}.
Based on solvers in the constraint-solving community, we carefully choose the abstractions to make syntactic constraints easier to express.


Altogether, the novel solver holds three noteworthy innovations:
\begin{enumerate}
    \item 
\solvername uses abstract syntax trees (ASTs) as its core abstraction to express constraints.
By having the constraints and the internal data structures of \solvername represented as ASTs, it is easier to represent and propagate syntactical constraints.
Our constraints are \textit{first-order}, meaning that nodes in the ASTs can be a range of values or variables.
Thus, a single first-order AST can represent a set of grounded ASTs, allowing for a compact representation.
The representational choice of first-order ASTs hence overcomes Popper's.

\item
\solvername builds upon ASTs to define \textit{simple program spaces}: a sub-space in which syntactic constraints are easier to enforce.
The power of constraint solving comes from working with problems of fixed structure (in terms of variables and constraints).
We innovate by introducing a key notion of what constitutes a fixed-structure problem in synthesis: a set of all programs represented with the AST of the same \textit{shape} and with a type associated with every node.
We term this structure \textit{uniform tree}.
\solvername \textit{lazily constructs uniform trees of increasing complexity as needed} and actively discards the ones that are no longer needed.

\item
\solvername introduces a set of concrete syntactic constraints together with their propagators. 
This set is fully extendable and provides basic functionality such as forbidding or enforcing certain sub-programs, as well as symmetry breaking. 
We further support the conjunction of arbitrary constraints.
Note that our ambition does not include introducing a standard constraint library for program synthesis, akin to the MiniZinc language or SMT-LIB.  
Instead, our work presents several constraints that we found useful for pruning symmetries and provides a general framework for integrating new constraints so that the community can build upon them.
While limited, this set already allows for the expression of a large portion of relevant syntactic constraints, as shown in the experiments.
We eventually provide proofs for the soundness and correctness of our propagators, i.e., we prove arc consistency~\cite{HandbookOFCP}.
\end{enumerate}

To demonstrate the benefits of using syntactic constraints for modelling program spaces and of a dedicated constraint solver, we evaluate our contribution on six program space enumeration tasks.
Specifically, we impose a set of intuitive syntactic constraints, such as those generated by Ruler~\cite{RULER2021Nandi}, that eliminate redundant programs, e.g., sorting a list twice is equivalent to sorting it once.
We then compare (i) how many programs the constraints eliminate, and (ii) how long it takes to enumerate all \textit{valid} programs.
We demonstrate that syntactic constraints eliminate a significant part of program space on standard benchmarks (eliminating up to 99\% of programs).
Indeed, because its constraint propagation avoids ever enumerating programs that violate constraints, \solvername enumerates program spaces between two and three magnitudes faster.
We further show the speed-up through first-order constraints, and how syntactic constraints can be utilized to guide existing state-of-the-art synthesizers, namely Probe~\cite{ProbeBarkePP20} and EUSolver~\cite{EUSolverSiYDNS19}.

\paragraph{Contributions and organisation}
In summary, the contributions of this paper to the literature are:
\begin{itemize}
    \item a definition of a \textit{constrained program space} which incorporates constraints over program structure into grammars;
    \item a set of constraints that help eliminate unwanted programs from program spaces;
    \item a constraint solver tailored towards program synthesis problems; and
    \item an integration of the solver into two common families of search procedures, top-down and bottom-up search; 
\end{itemize}

The remainder of the paper is structured as follows.
We first introduce a motivating example (Section~\ref{sec:motivating}) and revisit concepts in program synthesis and constraint programming, upon which we aim to build \solvername (Section~\ref{sec:background}).
In Section~\ref{sec:problem_statement}, we concretely define the term \textit{constrained program space} and formally state the problem we aim to solve. 
Section~\ref{sec:method_overview} provides an overview of and intuitions for our solver. 
Sections~\ref{sec:method:shaping} to~\ref{sec:search} are structured as follows:
We start by defining a language of constraints we want to express in Section~\ref{sec:method}. 
Second, we define the architecture and interface of \solvername and how to concretely propagate the constraints within simple program spaces in Section \ref{sec:method:solving}.
Third, we describe how to concretely implement top-down and bottom-up iterators using \solvername in Section~\ref{sec:search}.
We evaluate our method experimentally in Section~\ref{sec:experiments}.
In Section~\ref{sec:related_work}, we consider related work and conclude with a summary and potential future work in Section~\ref{sec:conclusion_future}.

\section{Motivating Example}
\label{sec:motivating}

In this section, we present a motivating example of how constraints can be a powerful tool to model program spaces.
We highlight constraints for effectively removing redundant programs and how to use them to guide the search. 

%

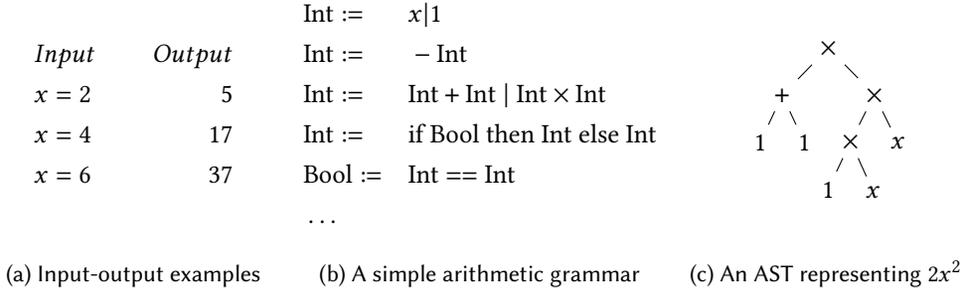
\begin{figure}[t]
    \begin{subfigure}{0.3\textwidth}
        \centering
         \begin{align*}
            &Input&Output&\\
            &x=2&5&\\
            &x=4&17&\\
            &x=6&37&\\
        \end{align*}
        \caption{Input-output examples}    
        \label{fig:LIA_specification}
    \end{subfigure}
    \begin{subfigure}{0.35\textwidth}
        \begin{align*}
            &\Int :=&& x | 1 \\
            &\Int :=&& -\Int \\
            &\Int :=&& \Int + \Int\ |\ \Int \times \Int \\
            &\Int :=&& \text{if Bool then Int else Int}\\
            &\Bool :=&& \Int == \Int\\
            &\dots&&
        \end{align*}
        \caption{A simple arithmetic grammar}
        \label{fig:LIA_grammar}
    \end{subfigure}
    \begin{subfigure}{0.3\textwidth}
        \centering
        \raisebox{5mm}{
        \begin{forest}
         [$\times$
            [$+$
                [$1$]
                [$1$]
            ]
            [$\times$
                 [$\times$
                    [$1$]
                    [$x$]
                ]
                [$x$]
            ]
         ]
        \end{forest}}
        \caption{An AST representing $2x^2$}
        \label{fig:LIA_program}
    \end{subfigure}
    \caption{A simple arithmetic program synthesis problem (LIA) and an example of a complete program. Figure (a) shows a specification by input-output examples, with Figure (b) describing the context-free grammar of possible programs. 
    Figure (c) shows a valid yet semantically redundant candidate program, as $1 \times x$ is equivalent to $x$.}
    \label{fig:LIA_example}
\end{figure}

\paragraph{Test-driven programming.}
Suppose that we are given a list of unit tests and need to find a function that passes them.
The unit tests define a list of integer inputs and test whether the function returns the desired integer outputs when executed on the respective input (\Cref{fig:LIA_specification}).
The programming language is rather simple and defines a range of operations over integers (\Cref{fig:LIA_grammar}).
This language is also called a linear-integer arithmetic (LIA) grammar.

\subsection{Using Syntactic Constraints to Remove Redundant Programs}
Given the seemingly simple problem, we want to quickly find a solution. 
However, the search space grows exponentially with the depth of the program.
Further, as the grammar contains many operators, the search space becomes quickly infeasibly large.
Fortunately, there are many redundant programs and symmetries. 
For example, \Cref{fig:LIA_program} contains the valid but redundant sub-program $1 \times x$.
Similarly, the operators $\Int + \Int$ and $\Int \times \Int$ are commutative, and thus a source of symmetry in the program space.
While holding semantic information, both symmetries can be enforced entirely on the syntactic level.
These syntactic constraints are easily discovered or given for many synthesizers, and straightforward to check on a syntactic level.

\paragraph{Simple search spaces.}
Notably, $+$ and $\times$ share a similar \textit{shape}: Both operate over integers and take two arguments. 
From a syntactic view, they behave exactly the same.
Hence, constraints enforced on the integer arguments of the addition are likely also applicable to the arguments of the multiplication.
We exploit this principle with the notion of \textit{simple search spaces}, that is, search spaces where, no matter the choice of operators, the \textit{shape of the program} does not change.

Simple search spaces make it easier to remove programs like $1\times x$, $1\times (x+1)$, \dots; programs that all follow a similar syntactic template $1\times \Int$.
Notably, we do not have to enforce this constraint everywhere, but only at the nodes in the AST where this template can occur, i.e., they \textit{match}.  
Matching template to AST nodes, our constraint system is centered around the speed and correctness of the matching function. 

\paragraph{Constraints over shapes.}
Instead of expressing a constraint on every program individually, we can express constraints over shapes, like $1\times \Int$.
We express these tree-like shapes using \textit{template trees}, which allow for two special types of first-order nodes.

First, we want to break the commutative symmetry for both $+$ and $\times$. 
We observe that many constraints apply on the same template or shape, just like for $+$ and $\times$.
Nodes in template trees can thus be a range of possible values they match. 
The template thus becomes $\Int \{+,\times\} \Int$. 
This increases inference strength in our simple program spaces:
Given $+$ and $\times$ are both possible at a node, we only have to break the symmetry once, not twice. 

Second, the if-then-else statement holds a different type of symmetry:
Here, we want the \textit{then} and \textit{else} branches of the statement to be different, but for any possible assignment. 
In other words, we want to match a wide range of possible programs and conditionally enforce a constraint. 
That is, we only care about the sub-program in the \textit{then} sub-branch if it matches the \textit{else} branch.

Formulating and enforcing these types of constraints for all possible programs is not easy.
Our language thus introduces a higher-order type of constraint, introducing variables to constraints that match a sub-tree within an AST.
First, this allows us to describe constraints over partial programs:
We can easily formulate and propagate constraints like \texttt{Forbid} $(1\times \Int)$ by using a variable \texttt{:A}.
Second, we can easily express the constraint on the if-then-else statement above, removing all programs that follow the shape \texttt{if Bool then :A else :A} with variable \texttt{:A}. 

\paragraph{A(n) (extendable) language for constraints.}
In the unit-test example, we want to forbid some shapes, like $1+\Int$, and enforce others, like ``the program must contain the input-symbol''.
We thus provide a range of standard operators for syntactic constraints we think are useful in practice. 
Beyond \texttt{Contains} and \texttt{Forbid}, i.e., a program must or must not contain this shape, many more operators are possible.
We express two more constraint operators: \texttt{Ordered}, used to break commutative symmetries, and \texttt{Unique}, expressing that the shape can occur at most once. 

This constraint language is not a standard language, like MiniZinc.
\solvername allows the formulation and description of additional constraints and their propagators.
Again, the notion of shape makes checking and propagating many constraints a lot easier.

\paragraph{Enforcing constraints on new operators.}
Assume that we have not found a solution program yet that passes all the unit tests.
We assume that this is due to the simplicity of the grammar. 
Thus, we add another operator to the grammar: \texttt{min(Int, Int)} and \texttt{max(Int, Int)}, which are, like the if-then-else statement, not given in the standard LIA formulation.
Similar to $\times $ and $+$, $min$ and $max$ are commutative, and $min(x, y)$ and $max(-x,-y)$ are the same. 
Though the concrete semantics of the operators are not defined in this SMT theory, we can easily come up with redundant programs and symmetries. 

\subsection{Using Syntactic Constraints to Guide the Search}
The second use of constraints we want to highlight is the modelling power they provide, beyond describing only feasible programs.

\paragraph{Defining useful programs.}
We are not restricted to formulating constraints that remove redundancy, but can describe constraints that are \textit{useful}.
For example, a \textit{useful} syntactic constraint is that every program should depend on the input. 
A program that does not contain the input is valid but likely useless.

Further in our example, we can directly see that the solution should be `simple' in nature.
Hence, we can choose to restrict the program space to simpler programs, introducing a \textit{language bias}.
For example, we can enforce that at most one if-then-else statement is contained in the target program, as we likely do not need that many conditionals. 
While both constraints do not remove \textit{redundant} programs, we can force the search to a useful sub-space.

\paragraph{Better search.}
Moreover, we argue that constraints are a useful tool for directing search.
Again, assume the case where we didn't find a solution and added the operators $min$ and $max$ to the grammar.
We ideally want to avoid re-exploring previously explored programs by ensuring that the new operators are present.

This resembles a crucial step in deploying program synthesizers: Finding the right domain-specific language (DSL), which involves adding and removing functionality from the grammar.
Our solver makes it possible to impose a constraint stating that \textit{every explored program should contain a new primitive}. 
This way, our solver \solvername only explores programs that have not been explored before.

\section{Background}
\label{sec:background}

To explain our contribution in detail, we start by introducing both program synthesis and the fundamentals of constraint programming.

\subsection{Program Synthesis}


Program synthesis (PS) is the task of deriving satisfying programs from (1) a specification, describing user intent, and (2) a program space, describing the space of possible programs. 
One way to express the user intent is through a set of input-output (IO) examples that should all be satisfied by the target program. 
The space of feasible programs is usually expressed by a context-free grammar, comprised of a set of derivation rules.
\begin{example}
    For the LIA domain, \Cref{fig:LIA_example} shows a specification by example and a context-free grammar (CFG) describing the program space.
\end{example}

Solving a program synthesis task is an enumerative search through the search space spanned by possible grammar rules. 
There exist multiple search strategies to find solution programs. 
We briefly introduce common terminology used for search.  

Here, a program is described by its \textit{abstract syntax tree} (AST). 
\textit{Top-down} search starts by taking the starting symbol of the grammar as the root.
Nodes corresponding to non-terminal symbols are called \textit{holes}.
Any program tree with one or more holes is called a \textit{partial program} and needs to be repeatedly expanded until there are no holes left, making it a \textit{complete program}.
Holes are filled using production rules, making it a concrete \textit{value node}.

\begin{example}
    The program 
    \begin{equation*}
        \texttt{if Bool then x+1 else Int}
    \end{equation*} 
    is partial, with two holes described by the non-terminals $\{\Bool,\Int\}$. 
\end{example}

Repeatedly applying derivation rules describes the \textit{search tree}, with program trees as nodes and connections between (partial) programs and their respective refinements. 
Note that the search tree is usually infinite and pruned by setting a maximum program depth and/or size.


\subsection{Constraint Programming}
\label{sec:CP}

Constraint Programming (CP) concerns 
formulating and solving Constraint Satisfaction Problems (CSP) or their optimisation variants \cite{CHIP-CP, MINI-CP}. 
Given a CSP, we try to find an assignment for a given set of variables that satisfies a given set of constraints. 
In the upcoming sections, we will review constraint programming techniques, so we can reapply them for the purpose of solving \textit{simple program spaces} in program synthesis.

\paragraph{Formulating CSPs}
A Constraint Satisfaction Problem consists of three components: a set of decision variables $X$, a set of their respective initial domains $D$, and a set of 
constraints $\mathcal{C}$.

A decision variable $x \in X$ is a variable that can take values in its corresponding domain $D(x) \in D$.  Whenever there is only 1 value a variable can take, that is $|D(x)| = 1$, we say a variable is \textit{fixed}.  A CSP is considered solved if all variables are fixed and all constraints are satisfied.

A constraint $c \in \mathcal{C}$ can restrict the feasible combination of variable values. 
For example, we may include a constraint $c = (x > y)$ to enforce that x must always be larger than y.
\begin{example}
    \label{ex:example_csp}
    \begin{equation*}
        \begin{aligned}
            X =& \,\{x, y, z\} \\
            D =& \,\{D(x), D(y), D(z)\} = \{\{1, 5, 7\}, \{2, 3, 5\}, \{-1, 1, 3, 5\}\} \\
            \mathcal{C} =& \,\{x \geq y, z = x - y\}
        \end{aligned}
    \end{equation*}
\end{example}
%
We can solve this example upon inspection and find a solution: $(x, y, z) = (5, 2, 3)$. 
Note that in this case, the solution is not unique. 

\paragraph{Solving CSPs}
Solving a given CSP consists of reasoning (propagation) and search (branching); some solvers also use clause learning and relaxations:
\begin{enumerate}
    \item \textit{Propagation:} Check for constraint violations and filter out impossible values from domains.
    \item \textit{Branching:} When no solutions can be filtered, split the problem into multiple sub-problems with smaller domains, such that any solution of the original problem can be found in the union of the sub-problems.
\end{enumerate}

We describe how constraint propagation is implemented in \texttt{Mini-CP} \cite{MINI-CP}, a CP solver for education purposes.
\textit{Propagation} is responsible for removing impossible values from domains according to the constraints. 
This process will continue until the constraints are unable to further reduce any of the domains. 
This is called a \textit{fixed point}: applying any propagation has no effect. 

To handle constraints, constraint-satisfaction solvers provide two functions:
\begin{enumerate}
    \item \textbf{Post constraint} is executed whenever the constraint is first imposed. 
    This is where the initial propagation of a constraint takes place.
    Variables, that occur in the constraint, add this constraint to their list of active constraints.
    \item \textbf{Propagate constraint} tries to shrink the domains of its related decision variables using a constraint type-specific filtering algorithm. Propagation gets triggered, e.g., on variable domain changes, and may trigger the propagation of other constraints.
\end{enumerate}

After executing the fix-point algorithm, we may not have found a solution yet. 
That is, there exists some variable $x \in X$ with $|D(x)| \geq 2$, and none of the constraint propagators can further reduce the domain size.
In that case, we \textit{branch} and divide the problem into two or more sub-problems such that all solutions of the original problem can be found in the union of the solutions of the sub-problems. Figure~\ref{fig:CPSimpleExample} demonstrates how a solver may branch on the domain of $x$. 

The repeated application of the branching scheme forms a tree structure.
The nodes of this tree are called states.
Exploring new states in the search tree is typically done using a depth-first traversal. 
This search strategy is memory efficient since only a single state needs to be maintained at a time \cite{MINI-CP}.  
Should a state (search node) be reached where one or more variables have an empty domain, then the solver backtracks to an earlier node in the search tree (e.g., the parent node), and chooses another sub-problem.  
Should the whole search tree be fathomed without finding a solution, then the problem is infeasible, having no solution.

\begin{figure}[tbh]
  \centering
  \hfill
  \begin{subfigure}{0.4\textwidth}
    \centering
    \includegraphics[width=\textwidth]{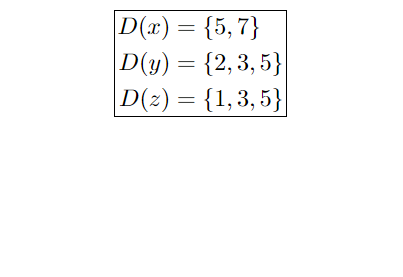}
    \caption{Fix-point in the root}
    \label{fig:CPSimpleExample2}
  \end{subfigure}
  \hfill
  \begin{subfigure}{0.4\textwidth}
    \centering
    \includegraphics[width=\textwidth]{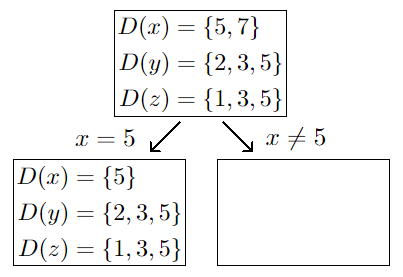}
    \caption{Branching}
    \label{fig:CPSimpleExample4}
  \end{subfigure}
  \hfill
  \caption{Solving the CSP from Example \ref{ex:example_csp} branching on the domain of $x$.
For the right branch, the constraint $x\neq5$ is posted, which triggers the propagation and the fix-point algorithm, leading to no feasible solution.}
  \label{fig:CPSimpleExample}
\end{figure}

\section{Problem Statement}
\label{sec:problem_statement}


In this section, we formally define the general problem of program synthesis and the specific problem of searching over a \textit{constrained program space}.
\begin{definition}[Program Synthesis]
    Given a context-free grammar $\G$ that specifies the syntax of the programming language and program space $\mathcal{P}_\G$, and a specification of user intent $\mathcal{S}:\mathcal{P}\rightarrow \{0,1\}$, \textbf{find} a program $p\in\mathcal{P}_\G$ derived from $\G$ such that $\mathcal{S}(p)=1$.
\end{definition}

We further define a constrained program space as:
\begin{definition}[Constrained Program Space]
    Given a context-free grammar $\G$, the original program space $\mathcal{P}_\G$ and a set of constraints $\mathcal{C}$, the constrained program space $\mathcal{P}_{\G,\mathcal{C}}$ is the set of all programs $p\in \mathcal{P}_\G$ such that $p$ satisfies all $c\in\mathcal{C}$.
\end{definition}

We call a (partial) program $p$ \textit{inconsistent} w.r.t. $\mathcal{C}$, if 
\begin{enumerate}
    \item $p$ violates any of the constraints $\mathcal{C}$, or
    \item all possible completions of $p$ violate any of the constraints $\mathcal{C}$.
\end{enumerate}

A context-free grammar (CFG) with rules dependent on their context in the form of constraints is called a context-sensitive grammar (CSG). 

\begin{definition}[Program Synthesis with Constraints]
    Given a \textit{context-sensitive} grammar $\G$ with constraints $\mathcal{C}$ and specification $\mathcal{S}:\mathcal{P}\rightarrow \{0,1\}$, \textbf{find} a program $p\in\mathcal{P}_{\G,\mathcal{C}}$ derived from $\G$ such that $\mathcal{S}(p)=1$.
\end{definition}

\section{Method Overview}
\label{sec:method_overview}



Given a program space and a set of constraints on that space, the main goal of our solver \solvername is to enumerate programs from the provided space such that no program violates the constraints.
The key conceptual idea 
is to iteratively and lazily construct \textit{simple program spaces} in which constraint propagation is efficient (\Cref{fig:method_overview}).
We formally define what a simple program space is later; intuitively, a simple program space consists of programs represented with an AST of the same shape and type.
In this context we use constraint solving and program enumeration as synonyms: we think of constraint problems as representations of program spaces such that every valid solution to the problem constitutes a program from that space.

Our synthesis framework effectively acts as a \textit{meta-solver} coordinating two solvers: a \textit{decomposition solver} formulating simple program spaces and a \textit{uniform solver} enumerating programs from the simple space.
The decomposition solver constructs problems and proactively propagates constraints that can be propagated early.
It also proactively detects when simple program spaces are invalid, i.e., all programs violate constraints.
The uniform solver iterates over all programs with the identical AST shape and makes sure that all returned programs satisfy constraints.  \solvername thus combines these two sub-solvers together with a search procedure, in order to achieve effective program synthesis.


\section{Modelling the Program Space (be)for(e) Synthesis}
\label{sec:method}
\label{sec:method:shaping}

In this section, we describe how to shape the search space. 
We first introduce a language to express constraints over ASTs, i.e., a language to constrain shapes.
For our set of basic constraints, we define the syntax and interpretation. 
Second, we formally introduce the notion of simple program spaces, that is, shapes in which constraints are easy to propagate and enforce.
Third, we intersect the two notions using \textit{local constraints}: constraints that hold for a certain simple shape.

\subsection{The Language of Constraints}
\label{sec:method:language}
We define a canonical language to describe syntactic constraints over ASTs, that is, the syntactic representation of a program.
Our constraint system is centred around the matching of (sub-)trees.
We start by introducing the different node types and defining matching.
Second, we introduce the different constraint operators, also called constraint types.

All constraints have the following shape.
\begin{definition}
Given a constraint operator $Op$ and a \textit{template tree} $t$ a constraint has the form
\begin{center}
	$Op$ $t$
\end{center}
\end{definition}
A template tree denotes all programs that fit this syntactic template.

The values of nodes in the template trees refer to the derivation rules in the grammar, used to derive a program.
To make sure the derivation rules are referred to unambiguously, we use the indices of the derivation rules as the values.
We use the grammar in Figure \ref{fig:LIA_grammar} to illustrate the constraints. 

The template trees are built from three types of nodes, which differ in the values they take.

\begin{definition}[Value node]
	A \textbf{value node} is a node with a value assigned to an index of a derivation rule from the given grammar $\mathcal{G}$.
	\label{def:valuenode}
\end{definition}

\begin{definition}[Domain node]
	A \textbf{domain node} is a node with a value assigned to \textit{a set} of indices from the given grammar $\mathcal{G}$.
	\label{def:domainnode}
\end{definition}

\begin{definition}[Variable node]
	A \textbf{variable node} is a node with a name assigned to it.
	\label{def:varnode}
\end{definition}

The domain node allows us to compactly represent a class of template trees.
The variable nodes allow us to unify/match a subtree and refer to it within a template tree.

\begin{example}
    We use curly braces to denote domain nodes and a starting colon to denote variable nodes. 
    The template tree in Figure \ref{fig:uniformsolver_input} contains domain node $\{+, \times\}$ and variable nodes $:a$.
\end{example}

\begin{figure}[t]
    \centering
    \begin{subfigure}{0.20\textwidth}
        \centering
        \begin{forest}
            [{\{$+$, $\times$\}}
                [$:a$]
                [$:a$]
            ]
        \end{forest}
        \caption{A template tree $t$}
        \label{fig:uniformsolver_input}
    \end{subfigure}
    \begin{subfigure}{0.38\textwidth}
        \centering
        \begin{forest}
            [,phantom
                [$+$
                    [$1$]
                    [$1$]
                ]
                [$\times$
                    [$+$
                        [$1$]
                        [$x$]
                    ]
                    [$+$
                        [$1$]
                        [$x$]
                    ]
                ]
            ]
        \end{forest}
        \caption{Concrete trees matching with $t$}
        \label{fig:uniformsolver_match}
    \end{subfigure}
    \begin{subfigure}{0.38\textwidth}
        \centering
        \begin{forest}
            [,phantom
                [$+$
                    [$1$]
                    [$x$]
                ]
            ]
        \end{forest}
        \caption{Concrete tree \textbf{not} matching with $t$}
        \label{fig:uniformsolver_nonmatch}
    \end{subfigure}
    \caption{Example for matches and non-matches for template tree $T$.
    The template tree in (a) enforces that both children are the same. Thus the tree in (c) does not match $t$.
    }
    \label{fig:uniformsolver_matches}
\end{figure}
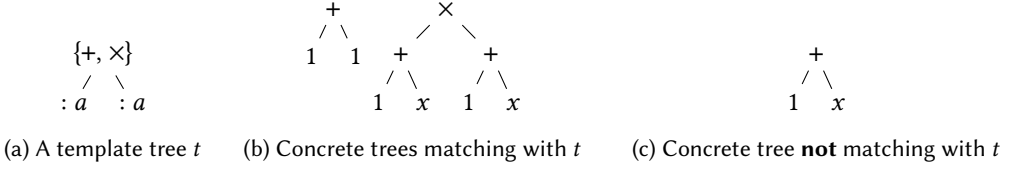

\begin{definition}
	Given a set of derivation rule indices $D$, the set of template trees $\mathcal{T}$ is recursively defined as:
	\begin{itemize}
		\item a value, domain, and variable nodes are template trees and in $\mathcal{T}$
		\item if children $c_1, \ldots, c_n$ are template trees, i.e. are in $\mathcal{T}$, then $(d, (c_1, \ldots, c_n))$ is also in $\mathcal{T}$ where $d \in D$ and $n$ is the arity of $d$
		\item nothing else is in $\mathcal{T}$. 
	\end{itemize}
\end{definition}

The constraints effectively forbid or require the presence of a template tree in the AST of programs.
We therefore define a matching operator between an AST and a template tree, which, in turn, depends on a matching operator between their nodes.
The constraint checking algorithms use the same matching operator to propagate the constraints. 

\begin{definition}[Node match]
    \label{def:node_match}
	A node $n_p$ from an AST matches a node $n_t$ from a template tree
	\begin{itemize}
		\item if $n_t$ is a value node, then $n_p$ and $n_t$ match if they have the same value
		\item if $n_t$ is a domain node, then $n_p$ and $n_t$ match if $n_p \in n_t$
		\item if $n_t$ is a variable node, then $n_p$ and $n_t$ match and, additionally, $n_t$ unifies with the tree rooted at $n_p$
	\end{itemize}
\end{definition}

\begin{definition}[Tree match]
    \label{def:tree_match}
	An AST $(a, (ac_1, ..., ac_n))$ matches a template tree $(t, (tc_1, ..., tc_m))$ iff $n = m$ and
	\begin{itemize}
		\item the root nodes $a$ and $t$ match, and 
		\item for all $i$ such that $1\leq i\leq n$, the subtree rooted at $ac_i$ matches the subtree rooted at $tc_i$.
	\end{itemize}
\end{definition}

\begin{example}
    The template tree $t$ in Figure \ref{fig:uniformsolver_input} matches with the trees in Figure \ref{fig:uniformsolver_match} but not with the tree in Figure \ref{fig:uniformsolver_nonmatch}.
\end{example}

Now that we have defined template trees and how to match them to ASTs, we introduce the set of constraint operators that can be applied to template trees.
The operators come from the following set of basic operators: \cstrname{Unique}, \cstrname{Contains}, \cstrname{Forbid}, \cstrname{Ordered}.
This is not an exhaustive list of all operators that might be useful for synthesis, but the ones that we have identified as useful for removing redundant programs. 
Our solver provides interfaces to easily add new operators.

We now define the constraints and their semantics.
First, the \texttt{Forbid} constraint can be used to forbid syntactically valid, yet semantically redundant sub-programs from the grammar. 
For example, the arithmetic grammar (see Figure~\ref{fig:LIA_example}) has a rule $Int := -Int$. 
A forbid constraint could be used to prevent $-(-(a))$ from appearing anywhere in the program tree. 
We do not lose the target program by eliminating such sub-programs, because $-(-(a))$ is represented by $a$ in another program.

\begin{definition}[Forbid]
	Given a program $p$ represented as AST and a template tree $t$, the program $p$ satisfies the constraint \texttt{Forbid} $t$ if there is no subtree $s$ of $p$ such that $s$ and $t$ match.
\end{definition}

The \texttt{Contains} constraint enforces that the provided template tree appears somewhere at or below the root. 
Within the arithmetic grammar, we could enforce that the program always contains $x$, actually depends on the input. 

\begin{definition}[Contains]
	Given a program $p$ represented as AST and a template tree $t$, the program $p$ satisfies the constraint \texttt{Contains} $t$ if there is a subtree $s$ of $p$ such that $s$ and $t$ match.
\end{definition}

The \texttt{Unique} constraint can be used to enforce that a certain grammar rule cannot appear more than once. 
In the LIA environment (see Figure \ref{fig:LIA_example}), this can be used to enforce that if-then-else is used only once to only search for simpler programs.

\begin{definition}[Unique]
	Given a program $p$ represented as AST and a template tree $t$, the program $p$ satisfies the constraint \texttt{Unique} $t$ if there is a subtree $s$ of $p$ such that $s$ and $t$ match and there is no other subtree $s^\prime$ of $p$, $s \not = s^\prime$, such that $s^\prime$ matches $t$.
\end{definition}

The \texttt{Ordered} constraint ensures that if the program tree contains the specified template tree, the matched variable nodes in the tree are ordered; by default, we use lexicographical order. 
This constraint is particularly useful for breaking commutative properties.
For example, in the arithmetic grammar, an ordered constraint can be used to ensure that only one of $a \times b$ and $b \times a$ is valid. 
Since they are semantically equivalent, we will not lose the target program by eliminating either one the two programs. 
The ordering we use is not important, as long as it is consistent throughout the search. 
We will use an ascending ordering in the rule index and tie break in a depth-first manner in case of equality. 

\begin{definition}[Ordered]
	Given a program $p$ represented as AST, a template tree $t$ with variable nodes $a,b,\dots$ and an ordering over trees $\leq_T$, the program $p$ satisfies the constraint \texttt{Ordered} $t$ with respect to $\leq_T$ if $t$ and $p$ match and the bound variable nodes in the tree follow $\leq_T$, i.e., $a\leq_Tb\leq_T\dots$. 
\end{definition}

By default, we set $\leq_T$ to be the lexicographical order over trees, but any total order comparing program trees is applicable.

\begin{definition}[Conjunction of Constraints]
    The program $p$ satisfies the \textit{set of constraints} $\mathcal{C}=\{Op_i\text{ }t_i\}$ if $p$ satisfies all $Op_i\text{ }t_i\in \mathcal{C}$.
\end{definition}


\begin{example}
    In the arithmetic grammar, the ordered constraint 
    \begin{equation*}
        \texttt{Ordered}\left( \vcenter{\hbox{\begin{forest}
            [{\{$+$, $\times$\}}
                [$:a$]
                [$:b$]
            ]
        \end{forest}}} \right)
    \end{equation*}
    can be used to ensure that only $a\times b$ and $a+b$ are explored, but $b\times a$ and $b+a$ are not.
\end{example}

\begin{example}
    The set of constraints 
    \begin{equation*}
        \left\{\texttt{Forbid}\left( \vcenter{\hbox{\begin{forest}
            [$\times$
                [$1$]
                [$:a$]
            ]
        \end{forest}}} \right),
        \texttt{Contains}\left( x \right)
        \right\}
    \end{equation*}

    allows program $p_1$ but eliminates program $p_2$ and $p_3$ of 
    \begin{equation*}
        p_1 = 
        \vcenter{\hbox{\begin{forest}
            [,phantom
                [$+$
                    [$1$]
                    [$\times$
                        [$x$]
                        [$x$]
                    ]
                ]
            ]
        \end{forest}}},
        p_2=
        \vcenter{\hbox{\begin{forest}
            [,phantom
                [$\times$
                    [$1$]
                    [$x$]
                ]
            ]
        \end{forest}}},
        p_3 = 
        \vcenter{\hbox{\begin{forest}
            [,phantom
                [$+$
                    [$1$]
                    [$1$]
                ]
            ]
        \end{forest}}}
    \end{equation*} 
\end{example}

\subsection{Simple Program Spaces}

So far, we have left the notion of a simple program space intuitive; we now precisely define it.

With simple program spaces, we want to represent a class of programs whose ASTs have the same shape.
Looking at it from a constraint-solving perspective, the nodes in the AST are the choice variables, where choices correspond to functions and arguments to be placed at the nodes, and no matter which choice we make at any of the nodes, the shape of the AST does not change.
If this is the case, constraint propagation can be done efficiently.

Consider how choices change the AST shape.
\begin{example}
    Consider the grammar from Figure \ref{fig:LIA_grammar}. 
    The addition rule $4$ has a different shape than the negation rule $3$, due to the different number of children.
    Rule $4$ has a different shape to the made-up rule 
    \begin{equation*}
        \Int := \texttt{negate\_on\_condition}(\text{Bool},\Int)
    \end{equation*}
    as the child's types are different.
    The addition rule $4$ has the same shape as the multiplication rule $5$.
    Hence, choosing between $4$ and $5$ leads to no changes in the shape of the AST.
\end{example}

We call the structure representing all programs with an identically shaped AST a \textit{uniform tree}.
A uniform tree can contain two types of nodes: value nodes and holes.
Value nodes are identical to \cref{def:valuenode}.

\begin{definition}[Hole]
    \label{def:hole}
	A \textbf{hole} is a node such that
	\begin{itemize}
		\item the value of a hole is a range of derivation indices of a grammar $\mathcal{G}$, and
		\item the children are holes or value nodes, or an empty set (indicating a leaf node).
	\end{itemize}
	\label{def:hole}
\end{definition}

Note that our notion of holes is different from the standard notion in the literature, where a hole is simply an unfilled part of an AST \cite{SKETCH,SyGuSAlur13,GulwaniPS17}.
In our interpretation, a hole is every node for which we did not make a final decision about its value.
Importantly, a hole can have children, which can also be holes or value nodes.

Now we introduce a special version of a hole: No matter the choice we make for the holes in a tree, the shape always stays the same.  
We call these \textit{uniform} holes.

\begin{definition}[Uniform holes]
    \label{def:uniform_hole}
	A \textbf{uniform hole} is a node such that
	\begin{itemize}
		\item its value is a range of derivation indices of a grammar $\mathcal{G}$,
            \item every value assignment from the range results in the same number of children
		\item its children are uniform holes or value nodes, and 
		\item each child has at most one type.
	\end{itemize}
	\label{def:uniformholes}
\end{definition}

The types of a hole are described by the types of its possible values. 
To keep its shape, we require holes to have at most one type, i.e., the same type of continuation.


Finally, we give a definition of a uniform tree.

\begin{definition}[Uniform tree]
    \label{def:uniform_tree}
	A uniform tree is a tree such that
	\begin{itemize}
		\item The root is either a value node or a uniform hole, and
		\item All children of the root are uniform trees.
	\end{itemize}
	\label{def:uniformtree}
\end{definition}

\begin{example}
    \label{ex:uniform}
    Consider
    \begin{equation*}
        p_1=\vcenter{\hbox{\begin{forest}
            [{\{$+$, $\times$\}}
                [{\{$1,x$\}}]
                [$-$
                    [{\{$ 1,x$\}}]
                ]
            ]
        \end{forest}}}, 
        p_2=\vcenter{\hbox{\begin{forest}
            [{\{$+$, $\times$\}}
                [{\{$+,x$\}}]
                [$-$
                    [{\{$ 1,x$\}}]
                ]
            ]
        \end{forest}}}
    \end{equation*}
    $p_1$ is uniform. $p_2$ is not uniform due to the hole on the left not being uniform, allowing for different shapes.
\end{example}

The uniform tree now represents a simple program space.
The goal for our solver is now to enumerate this program space and ensure that every program is valid, given constraints.

To accommodate that AST nodes are now a range of possible values, we update the definition of the matching function from \Cref{def:node_match}:

\begin{definition}[Node match for uniform trees]
    \label{def:node_match_uniform_trees}
	A node $n_p$ from a uniform tree, given by its domain, matches a node $n_t$ from a template tree
	\begin{itemize}
		\item if $n_t$ is a value node, then $n_p$ and $n_t$ match if $|n_p|=1$ and $n_t\in n_p$
		\item if $n_t$ is a domain node, then $n_p$ and $n_t$ match if $n_p \subseteq n_t$ 
		\item if $n_t$ is a variable node, then $n_p$ and $n_t$ match and, additionally, $n_t$ unifies with the tree rooted at $n_p$
	\end{itemize}
\end{definition}

The tree match for uniform trees remains the same.

\subsection{Local Constraints}
The constraints we introduced so far, such as \cstrname{Forbid} $t$, hold for programs \textit{globally}.
This is easy to check once a program has been generated.
However, when we actively try to propagate constraints, the solver might need to track several \textit{instances} of the same global constraint at different parts of a uniform tree.
We call such constraints \textit{local constraints}.

\begin{definition}[Local constraints]
	The \textbf{local constraint} of a constraint $c$ is a rooted version of $c$. 
	That is, local constraint holds a \verb|path| pointing to a location in the tree where $c$ applies.
	\label{def:localconstraint}
\end{definition}

\begin{example}
    \label{ex:local_constraints}
    Consider the uniform tree 
    \begin{equation*}
        p_1=\vcenter{\hbox{\begin{forest}
            [{\{$+$, $\times$\}}
                [{\{$1,x$\}}]
                [{\{$+$, $\times$\}}
                    [{\{$ 1,x$\}}]
                    [{\{$ 1,x$\}}]
                ]
            ]
        \end{forest}}}
    \end{equation*}
    The local constraint \texttt{Forbid} $1\times\Int$ rooted at path $[2]$, i.e., the right child of the root node, only propagates when any changes are made to the right sub-tree.
\end{example}

\section{Efficiently Propagating Syntactic Constraints in Program Spaces}
\label{sec:method:solving}

Having introduced the constraint language, program spaces, and local constraints, we are now ready to introduce our solver \solvername.
The design is comprised of two parts, the
\begin{itemize}
    \item \textbf{decomposition solver}, decomposing the program space into simple sub-spaces, and the 
    \item \textbf{uniform solver}, that propagates and iterates simple, or \textit{uniform}, program spaces.
\end{itemize}

We describe how to use this architecture to effectively propagate syntactic constraints. 
Subsequently, we introduce the specific propagation procedures for the constraints introduced in Section~\ref{sec:method:language} and prove their correctness.  Together with propagation, \solvername uses search, which is described in Section~\ref{sec:search}.

\subsection{Solver Architecture}


\begin{figure}[t!]
    \centering
    \includegraphics[width=0.98\textwidth]{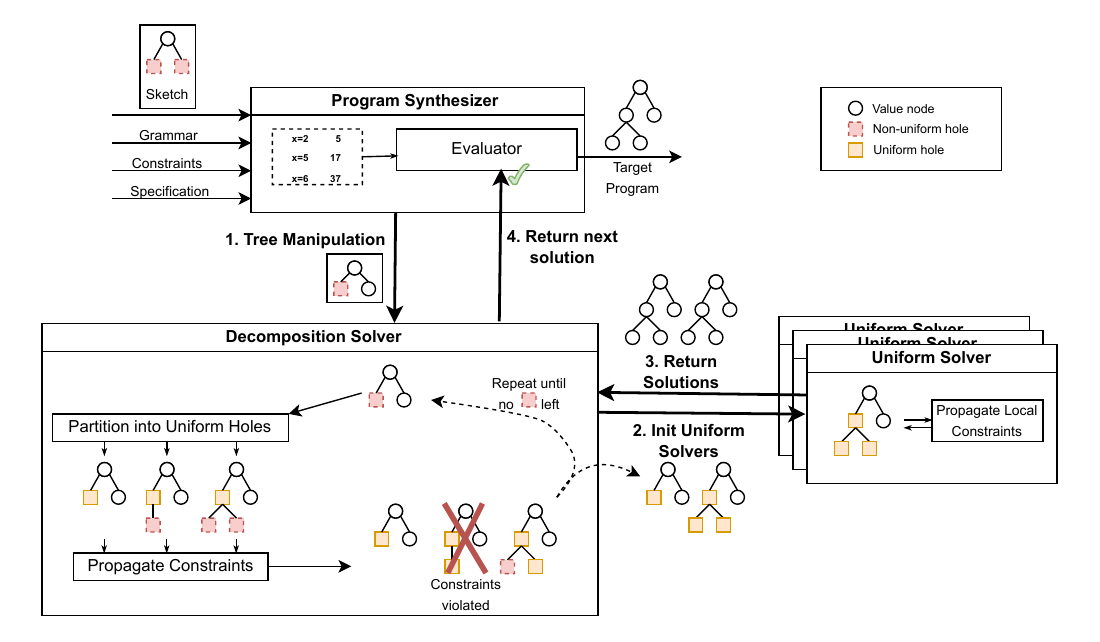}
    \caption{Overview of constrained program synthesis using \solvername. 
    (1) The synthesizer operates on the decomposition solver using tree manipulations.
    (2) The solver then propagates these changes and tries to iteratively construct uniform trees, i.e., simple program spaces.
    (3) For each uniform tree, a uniform solver is initialized, which propagates constraints and iterates solutions.
    (4) The decomposition solver forwards the solutions back to the synthesizer, where the solution is tested against the specification. 
    If satisfied, the target program is returned.}
    \label{fig:method_overview}
\end{figure}

\subsubsection{Decomposition solver}
\label{sec:decomposition_solver}

The \textit{decomposition solver} has a range of responsibilities.
Its primary task is to break down the program space into simple program spaces, i.e., uniform trees.
The second task is to maintain a valid state with which the search method interacts.
Note that the decomposition solver is not responsible for iterating over the program space; it only reacts to the command from a search procedure.
Thus, it is responsible for posting and propagating local constraints during the search.

We first describe the two decomposition algorithms: 1. how to partition a hole, and 2. how to decompose constraints into more powerful, \textit{local constraints}.
Second, we describe how to keep track of validity using the \textit{solver state}.
Third, we describe the interface with the solver state, i.e., the current program, limited to just three kinds of \textit{tree manipulations}.

\paragraph{Decomposing program spaces}
Given a hole (see Definition \ref{def:hole}), the solver tries to simplify its program space.
When called, the solver takes the hole and partitions it into the possible shapes it can derive to, i.e., uniform holes.
Once a tree is uniform, i.e., all holes are uniform, a uniform solver is initialized to solve the simplified problem, or simple program space.

\paragraph{Decomposing constraints}
Every \textit{global constraint} has a corresponding \textit{local constraint}.
The decomposition solver triggers whenever a new node (usually a hole) appears in the program tree. 
The solver will then post a local variant of the global constraint, propagating the constraint at that particular location.

Splitting up global constraints into local constraints has two main advantages:
\begin{enumerate}
    \item Local constraints can be temporarily deactivated. 
    By deactivating satisfied local constraints, we prevent checking these satisfied parts of a global constraint over and over again.
    \item We can reduce the frequency of unnecessary propagation. 
    On each tree manipulation, we can carefully choose for each active local constraint to either schedule it for propagation or ignore it.
    Thus, we do not need to check active constraints local to other branches in the program tree.
\end{enumerate}

To illustrate how global constraints are split up into local constraints, consider the grammar constraint \texttt{Forbid} ($1 \times a$) represented as a tree in Figure \ref{fig:grammarvslocalconstraints1}.
All nodes in this tree are new; thus, the solver posts a local constraint at each of the nodes in the tree. 
In the figure, each * represents a local variant of the forbid constraint at a particular node. 
During the fix point algorithm, all newly posted constraints are propagated to their respective location. 
Figure \ref{fig:grammarvslocalconstraints3} represents the state after propagation. 
3 out of the 5 local constraints are now satisfied and deleted. 
The other 2 constraints cannot deduce anything at this point and remain \textit{active}. 
This means that if a tree manipulation occurs at or below their path, they are scheduled for (re-)propagation.
\begin{figure}[tb]
    \centering
    \begin{subfigure}{0.22\textwidth}
        \centering
        \begin{forest}
        [$\times$
            [$1$]
            [$a$, fill=green!30,draw,diamond]
        ]
        \end{forest}
        \caption{Forbidden Tree. Variable node $a$ matches any sub-tree.}
        \label{fig:grammarvslocalconstraints1}
    \end{subfigure}
    \begin{subfigure}{0.38\textwidth}
        \centering
        \begin{forest}
        [{\{$+$, $\times$\}}, uhole, label={[label distance=-10pt]north east:*}
            [{\{$1$, $x$, $-$\}}, hole, label={[label distance=-10pt]north east:*}]
            [{\{$+$, $\times$\}}, uhole, label={[label distance=-10pt]north east:*}
                [$1$, label={[label distance=-10pt]north east:*}]
                [{\{$1$, $+$, $\times$\}}, hole, label={[label distance=-10pt]north east:*}]
            ]
        ]
        \end{forest}
        \caption{Local constraints are posted}
        \label{fig:grammarvslocalconstraints2}
    \end{subfigure}
    \begin{subfigure}{0.38\textwidth}
        \centering
        \begin{forest}
        [{\{$+$, $\times$\}}, uhole, label={[label distance=-10pt]north east:*}
            [{\{$1$, $x$, $-$\}}, hole]
            [$+$
                [$1$]
                [{\{$1$, $+$, $\times$\}}, hole, label={[label distance=-10pt]north east:*}]
            ]
        ]
        \end{forest}
        \caption{After propagation}
        \label{fig:grammarvslocalconstraints3}
    \end{subfigure}
    \caption{Forbid constraint (a) is imposed on tree (b) by posting a local constraint (*) at each location. 
    After propagating, one hole was filled. Only 2 of the local constraints remain active. The other 3 constraints are satisfied and thus deleted.
    Here, borders of uniform holes are orange, non-uniform holes are red and dashed, and variable nodes are green and diamond-shaped.
    }
    \label{fig:grammarvslocalconstraints}
\end{figure}
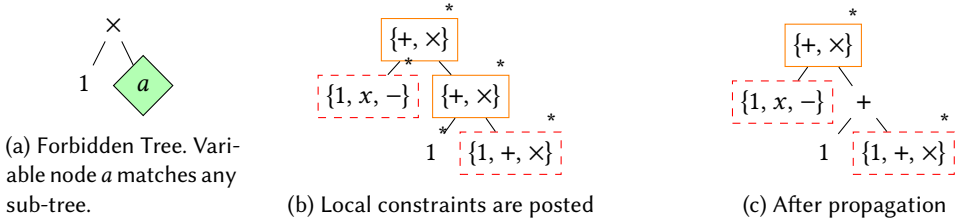

\paragraph{Solver state} 
A \textit{solver state} is a 3-tuple that fully describes the state of the solver.
It holds a (1) a partial program, (2) the list of \textit{active} constraints that could still be violated, and (3) a feasibility flag that indicates if the program still satisfies the constraints.
Note that the activation of local constraints is independent of any other solver state.
Local constraints are only valid for the current uniform tree and, thus, the current state.

After each tree manipulation, the solver will choose which constraints to schedule for propagation to ensure that the state satisfies all constraints. 
The constraints that need to be propagated depend on both the location of the hole and the type of constraint. 
On each tree manipulation, we will only schedule constraints that can potentially make any deductions.

\paragraph{Tree manipulations}
To be able to propagate constraints regardless of the type of search used, we make sure that all manipulations of the current state are made \textit{through the solver.}
Therefore, all search strategies manipulate the current state using a combination of primitive tree manipulations.

\begin{itemize}
    \item \texttt{remove!(solver, path, rule)}. Remove a rule from the domain of a hole at the given path. 
    If the remaining rules have the same shape, this hole will be converted to a uniform hole, and its children will be instantiated. 
    \item \texttt{remove\_all\_but!(solver, path, rules)}. Remove all rules from the domain of the hole at the given path, except for the specified remaining rules. If possible, the hole will be converted to a uniform hole.
    \item \texttt{substitute!(solver, path, new\_node)}. Substitute an existing node at the given path with a new node.
\end{itemize}
The solver is responsible for propagating relevant constraints after each of these tree manipulations.


\subsubsection{Uniform Solver}
The \textit{uniform solver} is for propagating constraints in a uniform tree and enumerating all valid programs.
Note that no new local constraints are posted, but only existing ones are propagated.

We describe our custom implementation of a uniform solver, supporting our custom constraint propagators.
Due to the nature of the simple program space, any existing solver can be used in its place.
Specifically, we implement an ASP version for our experiments.

The uniform solver closely follows the design of existing constraint solvers.
It performs a depth-first search over the assignments to the \textit{uniform holes} of the provided uniform tree. 
The solver uses a hole selection heuristic to pick a uniform hole, saves its state, and chooses a temporary value from the domain of the uniform hole.
Assigning a temporary value to one uniform hole triggers the fix-point constraint propagation algorithm.
This removes assignments from other holes that would result in an infeasible program.
The uniform solver then picks the next uniform hole until no more choices are left. 
This yields a complete program, and the solver backtracks to the last saved state, i.e., the last chosen uniform hole, and picks another value from its domain.
If no more values are left, then it backtracks to another uniform hole.
This exhaustively iterates all possible solutions of the uniform tree.

The program enumeration becomes memory-efficient by tracking changes since a saved state and reverting these changes to restore it. 
We use \textit{state sparse sets} \cite{MINI-CP} to describe domains, which allows us to easily revert to the previous state.



\subsection{Propagating Constraints}

In this section, we describe concretely how the constraints introduced in Section \ref{sec:method:language} are implemented and propagated. 
Note that we do not check constraints on complete programs, which is straightforward to do, but wasteful.
Instead, we aim to propagate the constraints to remove invalid programs before they are generated.
We provide pseudo-code in the supplementary material. 

Our propagators closely follow the design principles of constraint programming, which we tailor to our tree-shaped domain, i.e., uniform trees.
A core concept for efficient propagation is the notion of \textit{triggers}, which start a propagation and which we seek to minimize.
Here, the solver keeps track of triggers, i.e., the nodes relevant for a specific local constraint.
A local constraint is thus only propagated when a relevant node is updated.
Further, we want to minimize the number of constraints propagated by using the deactivation mechanism of local constraints.
Second, to maximize the inference strength of propagators, we aim to `localize' constraints as much as possible.
Thus, we want (1) constraints only to hold at a specific location, given by a path, and (2) to minimize the dependencies between sub-trees, as these dependencies must also be propagated.

\subsubsection{Forbid}
\label{sec:forbiddenconstraint}

Propagating a \cstrname{Forbid} constraint is straightforward with our notion of local constraints. 
First, the decomposition solver posts local versions of the constraint at all uniform holes that have the same shape, and thus could match the tree.

To propagate a local forbid constraint, given by a template tree and a path, the solver uses the pattern match function to match the forbidden tree with the node located at the path. 
The match describes whether the forbidden tree is in the set of programs, represented by the uniform tree. 
If the match fails, the local constraint is already satisfied and can be deactivated. 
If the match is successful, the forbidden tree is present in the program, so the state must be set to infeasible. 
If a match can be prevented by removing a rule from a hole, the constraint does so and then deactivates itself. 
It is also possible that multiple holes are involved, and no deduction can be made. 
In that case, the constraint remains active and will be re-propagated whenever one of the holes involved is updated.

\subsubsection{Contains}
\label{sec:containsconstraint}
The \texttt{Contains} constraint enforces that the program must contain a template tree.
Unfortunately, propagation is very expensive, as it requires pattern matching all nodes in the program tree with the template. 
However, in a uniform tree, we can deduce which nodes can match the \textit{shape} of the template tree, having to match type and children.
Thus, we only keep track of all potential matches, called \textit{candidates}.

Hence, we only need to match and update the candidates.
If a candidate matches the template, the constraint is satisfied and deactivated. 
If a candidate fails to match the template, it will be removed as a candidate. 
Finally, if there is only a single candidate remaining, it will be enforced to equal the template, and the constraint can be deactivated.

\subsubsection{Unique}
\label{sec:uniqueconstraint}
The \texttt{Unique} constraint is related to the \texttt{Contains} constraint, but enforces a derivation rule to appear at most once. 
Again, we only need to consider a list of candidates, i.e., nodes that match the shape.
When the local version is posted and propagated, we count the number of matches over candidates within the sub-tree.
If that number is $>2$, the current state is inconsistent. 
If that count is $1$, we post a local \texttt{Forbid} constraint to prevent it from occurring in any other branch. 
If the unique template tree only consists of a single node, we can directly remove that rule from any other domain.

\subsubsection{Ordered}
\label{sec:orderedconstraint}

An \texttt{Ordered} constraint is defined by a template tree with variable nodes and an order of these variable nodes.
Again, the decomposition solver posts a local \texttt{Ordered} constraint, rooted at a path, when the template tree matches the uniform tree.
This binds all variable nodes to node instances in the tree. 
The matched nodes are the triggers for this constraint.

Assume that we want to enforce $op(a, b)$ for variable nodes $a,b$ and order $a\leq b$ at a given path.
The constraint is enforced by iteratively comparing the rule indices of the nodes.
In case of equality, the tie is broken by comparing the children in a depth-first manner.

After the deductions, a \texttt{result} flag is returned that describes the current state of the $\leq$ inequality.
There are three possible results: 

\begin{itemize}
    \item Success. 
    $a \leq b$ is guaranteed under all possible assignments of the holes involved. 
    This means that the constraint is always satisfied and can be deactivated.
    \item Hard Fail. 
    $a > b$ is guaranteed under all possible assignments of the holes involved. 
    In this case, the constraint is violated, so the solver state must be set to infeasible.
    \item Soft Fail. 
    In this case, $a \leq b$ and $a > b$ are still possible depending on how the holes involved are filled. 
    We cannot make a deduction at this point, and the constraint needs to be re-propagated if one of the trigger nodes involved is updated.
\end{itemize}

Longer order chains, e.g., enforcing $op(a,b,c)$ with order $a\leq b\leq c$ can be reduced to two separate \texttt{Ordered} constraints, with $a\leq b$ and $b\leq c$.

\subsubsection{Conjunction of Constraints}
\label{sec:conjunction_of_constraints}
As \solvername is invariant to the order in which constraints are propagated (see Proofs in Section \ref{sec:methods:proofs}), the conjunction of constraints is propagated by propagating each conjunct individually.




\subsection{Proving Correctness}
\label{sec:methods:proofs}
\setcounter{theorem}{0}

We aim to prove the correctness of our constraint system. 
Specifically, we first show that our shape-based constraint formulation indeed captures and removes all unwanted solutions.
Second, we will use this result to prove the correctness of our propagators.

For this proof, we build upon two core properties of our constraint system.
First, our constraint system is itself centred around the matching function for shapes.
If a program matches the shape or potentially completes a match, the constraint is triggered. 
Second, within a simple program space, i.e., uniform trees, our constraint system is monotonic, that is, it only removes but does not add solutions to domains. 

Given a template tree $t\in\templatetrees$, a matching function $m:\templatetrees\rightarrow 2^\programs$ defines the subset of programs to which it matches.
We show the correctness of our recursive definition of the matching function (see \Cref{def:tree_match} and \Cref{def:node_match_uniform_trees}) and that the function $m$ matches exactly the programs we want.
Note that the programs $\programs$ do not need to be complete, but can have holes (see \Cref{def:hole}).
Further, we do not assume the programs $\programs$ to be uniform.

\begin{lemma}[Correctness of the Matching Function]
    \label{lemma:matching_function}
    Given a template tree $t$, a matching function $m:\templatetrees\rightarrow 2^\programs$, and a set of programs $\programs$:
    A program $p\in\programs$ matches template $t$ iff $p\in m(t)$.
\end{lemma}

\begin{proof}
    Let $t$ be a template rooted at node $r_t$, and $p$ be a program rooted at node $r_p$.
    We prove by structural induction, i.e., $r_p$ matches the template $r_t$ iff $r_p\in m(r_t)$.
    
    \textit{Case 1: $r_t$ is a value node.}
    $r_p\in m(r_t)$ if and only if $r_t = r_p$ by definition. 
    Then $r_p$ matches the template $r_t$.
    
    \textit{Case 2: $r_t$ is a domain node.}
    We prove by case for the domain compatibility of $r_p$ and $r_t$. 
    If $r_p$ and $r_t$ are disjoint, then $r_p\notin m(r_t)$ and $r_p$ does not match the template. 
    If $r_p \subseteq r_t$, then $r_p\in m(r_t)$, and all possible values match the template.
    Otherwise, there are elements in $r_p$ that are not in $r_t$, i.e., $r_p - r_t \neq \varnothing$. 
    Then by definition $r_p\notin m(r_t)$ and we need to remove elements from $r_p$ for a match.
    Thus, $r_p$ does not match the template.
    
    \textit{Case 3: $r_t$ is a variable node.}
    If the variable name of $r_t$ is not assigned, we assign $r_p$ to it.
    If the variable name $v$ of $r_t$ is assigned to a tree $t_v$, then we try to match $t_v$ with $r_p$.
    If $r_p\in m(t_v)$, then $r_p$ matches the template. 
    
    By the inductive hypothesis, if all children match recursively, then $p$ matches $t$ at the root as well. 
    
    Thus, $p \in m(t)$ if and only if $p$ matches the template $t$.
\end{proof}



Second, we prove the soundness of our constraint system with respect to a matching function.
Here, we call a solution to the constraint problem, i.e., a program, \emph{invalid} if it violates any of the constraints.

\begin{lemma}[Correctness of all Singular Constraints]
    \label{lemma:single_constraints}
    Given a constraint $C = Op\text{ }t$ for any
    $$
    Op\in\{Forbid, Contains, Unique, Ordered\},
    $$ 
    a template tree $t$, and a program $p\in\programs$:
    $p\in \programs$ is removed by the solver iff $p$ violates $C$.
\end{lemma}
\begin{proof}
    We proceed by cases over the types of constraints.
    By \Cref{lemma:matching_function}, we can use $m(t)$ to express the matches we want.
    
    \textit{Case 1: $C = Forbid\text{ }t$:}
    $C$ is violated if and only if there exists a sub-tree $p_i$ of $p$ that matches $t$, i.e., $p_i\in m(t)$.
    $p_i$ matches iff the constraint is triggered at a possible occurrence.
    Thus, $p$ is removed. 
    
    \textit{Case 2: $C = Op\text{ }t$ with $Op\in\{Contains, Unique\}$:}
    $C$ is satisfied if and only if there exist $n$ sub-trees $p_i$ of $p$ that matches $t$, i.e. $p_i\in m(t)$, where $n\geq 1$ for $Contains$ and $n\leq 1$ for $Unique$.
    As the shape has to match, any match must be at a candidate location (see \Cref{sec:containsconstraint}).
    The constraint is triggered if and only if a match happens at one of the candidate sites and the number of matches changes.
    $p$ is removed if and only if $n$ is violated.

    \textit{Case 3: $C = Ordered\text{ }t$, $t$ contains variables $a, b$ and enforces $a\leq_t b$ for some tree ordering $\leq_t$:}
    Assume that $a,b$ were matched to sub-trees $t_a, t_b$.
    By definition, $C$ is violated if and only if $t_a \nleq_t t_b$, i.e., any possible complete program within $t_a, t_b$ violates the constraint.
    $p$ is removed if and only if $t_a \nleq_t t_b$.
    The $Ordered$ constraint is triggered at initialization of the trees, and if $t_a, t_b$ change through any tree manipulation.
    Thus, at all times, $p$ is removed iff $p$ violates $C$.
\end{proof}

This proves the arc-consistency~\cite{HandbookOFCP} of the individual constraints.
Together, the lemmas let us prove our main theorem, proving arc-consistency for a set of constraints:

\begin{theorem}[Correctness Theorem]
    \label{theorem:correctness}
    Given a conjunction of constraints $\mathbf{C}=\{C_1, C_2, \dots\}$, 
    a template tree $t$, and a program $p\in\programs$:
    $p\in \programs$ is removed by the solver iff $p$ violates $\mathbf{C}$.
\end{theorem}
\begin{proof}
    It suffices to show that the order of propagation does not influence the set of eventual solutions.
    $\mathbf{C}$ is violated if and only if any of its conjuncts are violated.
    Assume $p$ violates constraints $C_1, C_2,\dots$, and $p$ is correctly removed by enforcing $C_1$ using \Cref{lemma:single_constraints}.
    By the monotonicity of the constraint system, $p$ is not added back to the set of possible solutions. 
    Thus, $p$ cannot violate any other constraint.
\end{proof}

\section{Searching Shaped Program Spaces}
\label{sec:search}

\solvername only describes the space of possible programs, but relies on a search method to define how to iterate over it.
Our solver easily integrates into existing search procedures in program synthesis.
We demonstrate how this can be done for the families of top-down and bottom-up search.
In both cases, the concept is the same: both search procedures use the decomposition solver to construct uniform trees instead of individual ASTs.

\textit{Top-down search} maintains a priority queue of programs explored so far.
Initially, it contains only a decomposition solver state made up of a uniform hole containing the starting symbol and the global constraints.
Top-down search iteratively pops the first element in the priority queue and processes it.
At this point, two things can happen.
If the retrieved item is a \textit{uniform tree}, the search procedure retrieves the next complete program from it using the uniform solver and its hole heuristic, and enqueues the uniform tree back to the priority queue. 
Note that, depending on a priority function, the uniform tree can take arbitrary places in the priority queue. 
If the retrieved item is a \textit{decomposition solver state}, then the search procedure finds a non-uniform hole and refines it into uniform holes.
For each uniform hole, it creates a new state and enqueues it back into the priority queue.
At this point, constraint propagation is triggered, and all constraints that could be propagated are propagated. 
If there are no more non-uniform holes, the search procedure enqueues a uniform tree.

The choice of the hole heuristic and priority function defines the type of search.
Enqueuing trees at the top position implements a depth-first search, whereas enqueuing at the bottom implements breadth-first search.

\textit{Bottom-up search} iteratively combines programs from a bank of programs to construct bigger programs.
In our case, the bank contains uniform trees.
The initial bank contains uniform trees representing the terminal of a given grammar, i.e., all terminals of the same time end up in the same uniform tree.
Then, to construct bigger programs, the search procedure creates uniform holes for non-terminal derivation rules, i.e., all non-terminal rules of the same type and the same child types end up in the same uniform hole.
The search procedure then fills the chosen uniform hole with the uniform trees from the bank.
The resulting trees are, again, uniform trees from which complete programs can be enumerated.
The filling of the holes triggers constraint propagation, ensuring that every program retrieved from the uniform tree is valid.

\section{Experimental Evaluation}
\label{sec:experiments}

We evaluate our proposed solver by answering the following empirical research questions:
\begin{itemize}
    \item[\textbf{RQ1:}] How much of the program space needs to be pruned to justify the overhead of constraint propagation in program synthesis?
    \item[\textbf{RQ2:}] How much do state-of-the-art synthesizers benefit from modelling the program space on competition benchmarks?
    \item[\textbf{RQ3:}] Does combining multiple grounded constraints into a single first-order constraint improve the performance of the search?
\end{itemize}


\paragraph{Tasks}
We experiment over a set of diverse kinds of program synthesis domains:
\begin{itemize}
    \item\textbf{Arithmetic:} 
    The arithmetic grammar is closely related to the LIA grammar \cite{SMTlibBarrettMRST10}. 
    It is used to derive simple arithmetic expressions over the input symbol $x$ (see Figure \ref{fig:LIA_example}), but restricted to simple operations like $+, \times, -$. 
    Constraints are used to break symmetries (see \Cref{sec:motivating}).
    \item\textbf{Robots:} 
    The Robots environment \cite{BruteCropperD20} 
    consists of an $n \times n$ grid, and positions of the robot, a ball, and the goal. 
    The task is to learn to move the robot to the goal and let it grab and drop the ball. 
    Constraints are used to remove the many redundant routes, e.g., moving right directly after moving left.
    The grammar consists of 10 rules. 
    Even though this problem stems from a planning domain, we use it to highlight the power of syntactic constraints over black-box operators.
    Being stateful, the Robots domain is much easier to express in, e.g., Julia, than in SMT.
    \item\textbf{Symbolic:} 
    This semantics-free problem is used to push the solver to its limits, containing many operators of the same shape.
    Semantic grammars can be constrained to prune a large amount of the program space, but a significant amount of it remains valid. 
    In traditional CP, often only a few, if any at all, solutions exist. 
    Most of the time is spent finding a solution to the constraints.
    To mimic this setting and push constraint propagation to its limit, we will consider the highly constrained symbolic grammar, in which only a handful of programs satisfy the constraints.
    \item\textbf{Lists:} 
    The list grammar is commonly used for benchmarking in program synthesis~\cite{DeepcoderBalogGBNT17,NeoFengMBD18,RULER2021Nandi} and supports basic list operations and can be used to construct a program that takes two input integers $x$ and $y$ and returns a list. 
    Constraints are derived from Ruler~\cite{RULER2021Nandi}, e.g., reversing a list twice.
    \item\textbf{SyGuS SLIA and BV:} 
    A set of more difficult string-manipulation (SLIA) and bit-value-manipulation problems from the SyGuS Challenge 2019 \cite{SyGuS2019}, commonly used to benchmark synthesizers.
\end{itemize}

All concrete grammars and constraints used, are in the Supplementary Material.

\paragraph{Implementation details.}
All the experiments have been executed on an Intel i7-10750H CPU @ 2.60GHz with 16 GB of RAM, and were implemented in \verb|Herb.jl v0.6.0|\footnote{https://github.com/Herb-AI/Herb.jl/releases/tag/v0.6.0}.

\subsection{Reducing the Program Space (RQ1)}
\label{sec:exp:reducing_programspace}
We study (1) how syntactic constraints help to reduce the search space, and (2) whether syntactic constraints can remove many programs before exploring them.
Note that we do not execute a single generated program here, but only consider the size of the program space.

We will use top-down search to enumerate all programs in the program space of the Arithmetic, Robots, Symbolic, and Lists grammar up to a maximum depth. 
We omit the SyGuS grammars, SLIA and BV, because they (1) change between problems, and they (2) are too big to enumerate all possible programs naively.
We explore program space reduction on SLIA and BV in \Cref{sec:exp:guiding_synthesis}.
To measure the correctness and effectiveness of the constraint propagation, we will compare three variations of enumeration.
\begin{enumerate}
    \item \textbf{Plain Enumeration (+checking)}. Enumerates all programs, ignoring the constraints. Then, retrospectively checks the constraints and eliminates all programs that violate any constraint.  This is `generate and test' in CP terminology.
    \item \textbf{Plain Enumeration}. Enumerates all programs, ignoring the constraints.
    \item \textbf{Constrained Enumeration}. Enumerates all programs that satisfy the constraints using the proposed constraint solvers.
\end{enumerate}

For all experiments and using the exact same constraints, the programs obtained from methods 1 and 3 are exactly the same.
This means constraint propagation does not eliminate valid programs, nor keep any invalid programs. 
Although this is not a hard proof of correctness, it does increase the confidence that the propagators are implemented correctly. 

In \Cref{fig:RemainingProgramSpace}, we compare the program space with and without constraints by dividing the number of valid programs by the total number of programs using top-down enumeration. 
We see that the constraints significantly reduce the program space.\footnote{The symbolic grammar is omitted from this graph, since its smallest valid program has 12 nodes.}
With a maximum of 11 nodes, the imposed constraints can already eliminate roughly $99\%$ of the total program space for the List and Robots grammar.
Further, the runtime of enumeration reduces significantly using constraints:
For enumerating all programs with 12 AST nodes, search time is 3000, 123, 12, and 2.5 times slower for the Symbolic, Robot, Arithmetic, and Lists grammar, respectively.
The full results can be found in the Supplementary Material. 

\begin{figure}[ht!]
    \begin{subfigure}{0.49\textwidth}
        \centering
        \includegraphics[width=\textwidth]{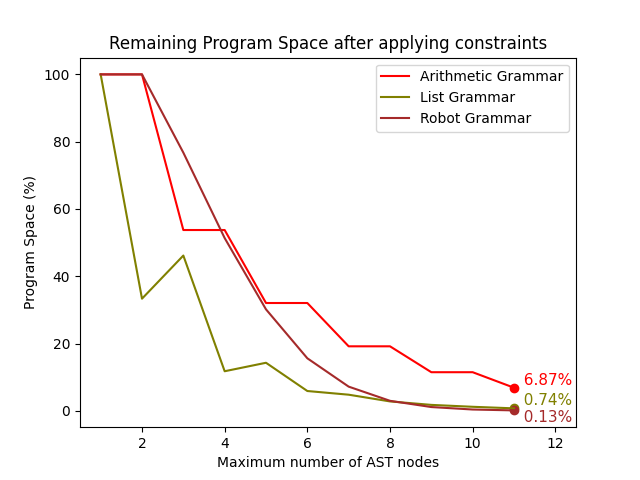}
        \caption{Program Space}
        \label{fig:RemainingProgramSpace}
    \end{subfigure}
    \begin{subfigure}{0.49\textwidth}
        \centering
        \includegraphics[width=\textwidth]{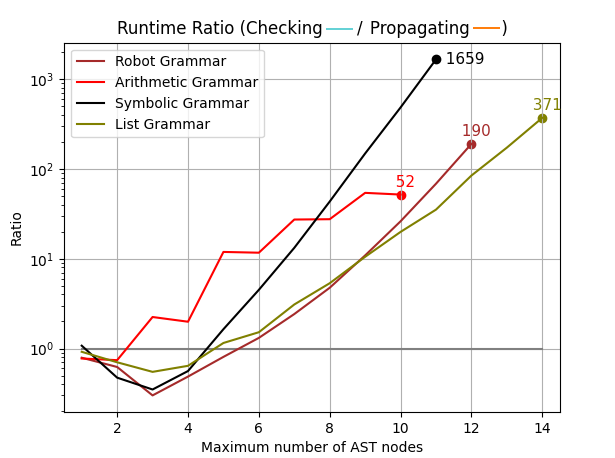}
        \caption{Runtime Ratio Checking vs. Propagating}
        \label{fig:ProgramEnumerationRuntime}
    \end{subfigure}
    \label{fig:spacereduction}
    \caption{(a) The remaining portion of the program space after applying the constraints, represented by the ratio of valid/total programs in the grammar; 
    (b) Runtime ratio of top-down program enumeration checking retrospectively to our propgation system. 
    We observe that checking is significantly slower than our propagation of syntactic constraints.}
\end{figure}

\Cref{fig:ProgramEnumerationRuntime} compares the runtime of enumeration checking constraints in comparison to propagating constraints.
We observe that the runtime of plain enumeration checking retrospectively the constraint is strictly higher than that of plain enumeration alone. 

Propagating the constraints during top-down search outperforms plain enumeration. 
This is an expected result, as only a small fraction of the program space needs to be enumerated in a constrained search. 
The ratio plot reveals that for all four grammars, the constrained search performs significantly better as the program space grows larger, but the exact improvement highly depends on the grammar.
An interesting observation is that the plain enumeration outperforms constrained enumeration for small program spaces. 
In these cases, the reduced program space does not justify the overhead of propagating constraints. 
We conclude that when a relatively large portion\footnote{Roughly speaking, $75\%$ or more. 
This depends on many other factors, such as the grammar and the type and amount of constraints.} of the program space can be eliminated, constraint propagation outweighs its overhead.


We also compare methods 1 and 3 for the bottom-up search. 
As the program space is exactly the same as for top-down enumeration, we compare the bank size during the enumeration of the Lists grammar.
Naturally, we can significantly reduce the bank size by expressing groups of programs as uniform trees.
For example, we require only 1318 uniform trees to represent 45,594 programs in the bank for a program depth of 5.
As programs from the bank could satisfy a constraint with an added parent, our current constraints cannot remove any (partial) solutions from the bank. 

\subsection{Guiding State-of-the-Art Synthesizers (RQ2)}
\label{sec:exp:guiding_synthesis}
A core motivation for using syntactic constraints is modelling the program space to guide the search towards \textit{useful} programs.
Hence, to answer RQ2, we evaluate two state-of-the-art synthesizers, namely Probe \cite{ProbeBarkePP20} and EUSolver \cite{EUSolverSiYDNS19}, on two common program synthesis benchmarks: The SLIA and the BV benchmark from the SyGuS challenge \cite{SyGuS2019}.
We compare both (1) without constraints, (2) with constraints removing redundancies, and (3) removing redundancies and with \textit{useful} constraints.
Not removing redundancies but still useful, we enforce that all programs must contain the input symbol in steps (2) and (3).

Probe and EUSolver also define a notion for useful programs. 
Both leverage sub-optimal solutions, i.e., solutions that solve some but not all given input-output examples.
Probe prioritizes grammar rules that were part of partially successful programs that were found earlier.
By contrast, following a divide-and-conquer approach, EUSolver tries to combine partially successful programs into a larger successful program.

We use constraints to nudge both synthesizers to find useful programs quicker by constraining the search space.
For SLIA, we follow our motivating example: 
The grammars are usually built around a core grammar, adding specific operators for the specific problem.
We enforce that the solution must contain all these added rules.
For BV, we follow another idea:
We observe that both approaches benefit from `diversity' in their selection of rules. 
That is, we want both approaches not to overuse the same operator.
We thus enforce that operators cannot be used twice in a row, e.g., removing $\Int + (1 + \Int)$, but not $(\Int + x) * (1 + \Int)$.
Note that we can choose to remove constraints during search, for example, whenever a new uniform tree is initialized or between cycles of synthesizers.

We base EUSolver on a breadth-first search.
Further, we enumerate SLIA and BV for $10^6$ ($3\times 10^6$ for Probe due to multiple cycles) and $10^5$ programs.

\begin{figure}[ht!]
    \begin{subfigure}{0.49\textwidth}
        \centering
        \includegraphics[width=\textwidth]{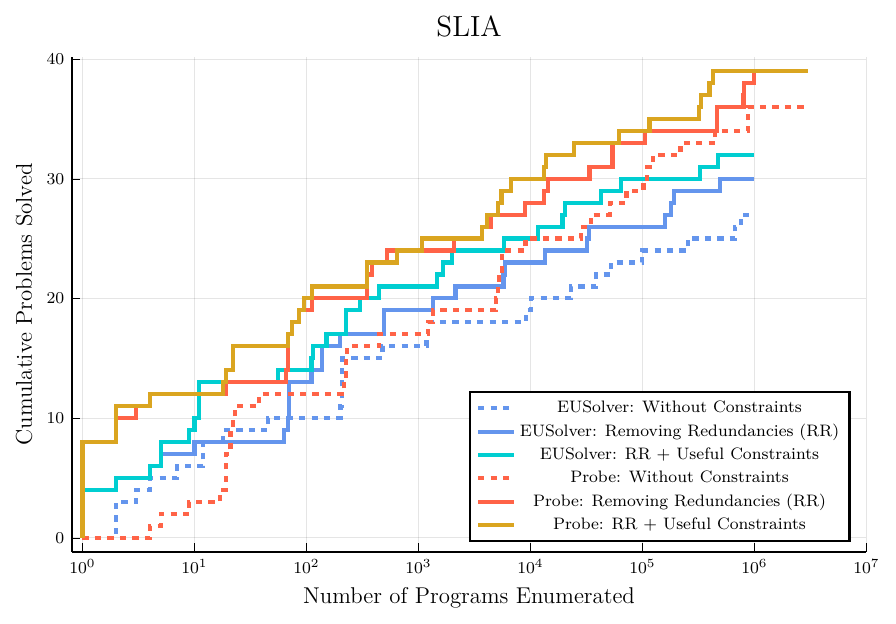}
        \caption{String manipulations}
        \label{fig:SLIA_results}
    \end{subfigure}
    \begin{subfigure}{0.49\textwidth}
        \centering
        \includegraphics[width=\textwidth]{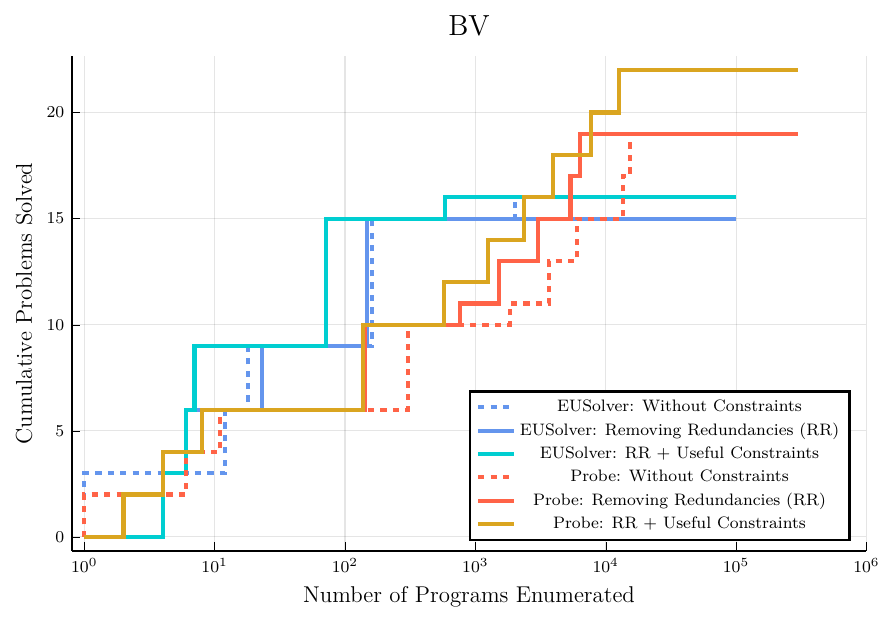}
        \caption{Bit-value manipulations}
        \label{fig:BV_results}
    \end{subfigure}
    \caption{Results of the synthesizers Probe and EUSolver on the (a) SLIA and (b) BV benchmark from the SyGuS challenge using no constraints, with constraints removing redundancies, and adding \textit{useful} constraints. 
    Useful constraints can be a tool to nudge synthesizers towards finding solutions faster. 
    }
    \label{fig:SOTA_reduction}
\end{figure}

The results are shown in \Cref{fig:SOTA_reduction}.
We see that both approaches benefit from the added constraints and find solutions faster.
Further, the simple, useful constraints let both approaches find shallow solutions even sooner.
We conclude two things:
First, formulating useful constraints in our language in real-world competitive program synthesis benchmarks is straightforward.
Second, even simple constraints have large impacts on the program space, and thus a better search.

\subsection{First-order Constraints (RQ3)}
\label{sec:exp:first-order_constraints}

We introduced two kinds of first-order nodes that can be used to constrain a class of sub-trees: domain and variable nodes. 
We analyze the impact of using first-order constraints on the propagation to answer RQ3.

We first test the effectiveness of combining constraints using domain nodes using the List grammar, a common program synthesis grammar. 
Enumerating all programs up until depth 8, we evaluate two settings: 
Either using the two first-order constraints or 12 equivalent grounded constraints. 
We report the run-time and propagations enumerating the entire program space.

\Cref{fig:FirstOrdervsGroundedListGrammar} shows that having first-order constraints means fewer constraint propagations need to run. 
We also see an improvement in runtime, which indicates that a higher quantity of grounded propagators is more expensive than a lower quantity of first-order propagators. 

\begin{figure}[h!]
    \centering
    \begin{tabular}{lrrrr}
        & \multicolumn{1}{l}{Programs} & \multicolumn{1}{l}{Runtime} & \multicolumn{1}{l}{Search Nodes} & \multicolumn{1}{l}{Propagate Calls} \\ 
        \hline
        First-order Constraints & 1 779 631 & 15.498s & 2 263 318 & 2 626 046 \\
        Grounded Constraints & 1 779 631 & 20.187s & 2 263 318 & 3 723 245 \\
        \hline
    \end{tabular}
    \caption{Enumerating programs of the List Grammar with a maximum program size of 15 using first-order and grounded constraints.}
    \label{fig:FirstOrdervsGroundedListGrammar}
\end{figure}

Second, we consider the Symbolic Grammar to measure the effectiveness of domain.
Again, note that this grammar contains no semantics and is used to push the solver to its limits.
In two experiments, we respectively introduce one of the following \texttt{Forbid} constraints:

\begin{equation*}
        \texttt{Forbid}\left( \vcenter{\hbox{\begin{forest}
            [{$\{b1,b2,b3\}$}
                [{$\{t1,t2,t3,t4\}$}]
                [{$\{t1,t2,t3,t4\}$}]
            ]
        \end{forest}}} \right)
\end{equation*}
Both constraints eliminate all trees that contain any of the $3\cdot4\cdot4 = 48$ forbidden sub-trees. 
Alternatively, we can break down the first-order constraint into $48$ grounded constraints, each forbidding a specific sub-tree.
Each individual constraint must be propagated and enforced individually.

We run the constrained program enumeration of the Symbolic Grammar up to a maximum of 8 AST nodes, and vary the number of constraints we use. 
In each run, we will use 1 first-order constraint and 1 grounded constraint for each missing case. 
For example, suppose we remove one of the rules from a domain node.
Then, $4\cdot 4 = 16$ grounded constraints must be constructed to cover the missing cases: 

\begin{equation*}
        \texttt{Forbid}\left( \vcenter{\hbox{\begin{forest}
            [{$\{b1,b2,\xcancel{b3}\}$}
                [{$\{t1,t2,t3,t4\}$}]
                [{$\{t1,t2,t3,t4\}$}]
            ]
        \end{forest}}} \right)
        \longrightarrow
        \texttt{Forbid}\left( \vcenter{\hbox{\begin{forest}
            [{$\{b3\}$}
                [{$\{t1\}$}]
                [{$\{t1\}$}]
            ]
        \end{forest},
        \begin{forest}
            [{$\{b3\}$}
                [{$\{t1\}$}]
                [{$\{t2\}$}]
            ]
        \end{forest},\dots
        }} \right)
\end{equation*}
 
Note that the number of valid programs remains the same. 
The difference lies in the number of propagators and inference strength. 
Regardless of which combination of constraints is used, we will always enumerate 1,358,656 out of the 2,355,328 total programs.

However, we see a positive correlation between the number of constraints and propagations.
For example, with 10 constraints, \solvername requires $\approx$200,000 propagation calls, whereas with 50 constraints, more than 3,000,000 propagations are required.
We observe a similar trend for the other tested grammars.
The full results are shown in the Supplementary Material.
We conclude that a higher quantity of grounded propagators is computationally more expensive than a lower quantity of first-order propagators.

We also see a positive correlation between the number of constraints and search nodes.
This can be explained by inference strength. 
When a first-order constraint is split up into grounded constraints, it can no longer exploit constraint interaction. 
For example: $\texttt{Forbid}(\{b1, b2\})$ is able to deduce that the domain $\{b1, b2\}$ is inconsistent. 
However, the individual constraints $\texttt{Forbid}(b1)$ and $\texttt{Forbid}(b2)$ separately cannot make any deductions towards inconsistency.
They need to enumerate the solutions within the tree before they can spot the same inconsistency.

Testing the influence of variable nodes in constraints, we use constraints on the Arithmetic grammar.
All constraints for this grammar are formulated using a variable node, e.g., \texttt{Forbid(0+\Int)} or \texttt{Ordered(:a + :b), [:a,:b]}.
We see that variable nodes allow for a brief and concise formulation of many intuitive constraints. 
Grounding out all constraints is naturally much slower.
For example, enforcing \texttt{Forbid 0+:a} at the root grounds out to 184,632,701 constraints for a maximum size of 12 AST nodes.
Each constraint is propagated separately.
More interesting is the inference power we gain over checking, shown in \Cref{fig:ProgramEnumerationRuntime}.
We see that many programs can be removed early, on leading to a significantly reduced program enumeration time.

We conclude that bundling grounded constraints into a first-order constraint does increase the strength of the propagation inference.

\section{Related Work}
\label{sec:related_work}

\paragraph{Encoding synthesis into constraints.}
A range of earlier program synthesizers, such as Sketch \cite{Sketch2009Solar-Lezama}, Brahma \cite{Brahma2011Gulwani}, SyPet \cite{SyPet2017Feng}, and Rosette \cite{rosetteTorlakB13,rosetteTorlakB14} encode the program synthesis problem in a constraint satisfaction problem.
A constraint solver is then employed directly to find the solution program. 
Sketch and Rosette can only make this translation by bounding the program space, and constraining the set of allowed operators.
While these approaches encode program spaces defined by grammar into constraints, we go further and model a program space beyond its grammar.
We further provide a clear interface for developing custom constraint propagators on the level of ASTs, and in \solvername demonstrate how it supports better constraint solving tailored towards program synthesis.
Moreover, these earlier approaches require an encoding of programming language semantics within constraints, whereas our solver does not. 

\paragraph{Solving uniform trees with existing constraint solvers.}
Constraint solving on the level of uniform trees also makes it easier to deploy existing constraint solvers more effectively.
Our goal is to gradually add support for using existing solvers, so that the users can flexibly switch between constraint solvers to enumerate uniform trees.
We expect that existing SMT-solvers like Z3 \cite{z3} and CVC5 \cite{cvc5}, and ASP-solvers like Clingo \cite{clingo} will be more efficient at certain types of programs and constraints.
However, our solver provides a clear interface for implementing custom constraint propagators, that makes adding new propagators and constraints straightforward.

\paragraph{Popper and solver-as-program-enumerator.}
Popper \cite{POPPER} is the closest related work to ours and gives direct inspiration.
Popper uses an ASP solver \texttt{clingo} to enumerate programs from a defined program space, i.e., it specifies a constraint model such that every solution is a valid program.
An innovative part about Popper is that every explored program, unless it is the solution program, yields a new constraint that further restricts the program space. 
Therefore, Popper, like us, models the program space beyond the grammar restrictions.
In contrast to our approach, which supports defining and constraint propagation over arbitrary program space, Popper \emph{only supports the logic programming family} of programming languages.
Moreover, the set of constraints Popper induces while searching for a program is pre-defined; our solver allows users to formulate arbitrary first-order constraints.

\paragraph{Semantics-guided synthesis.}
A semantics-guided synthesiser uses constraints in the form of constrained horn clauses \cite{chc} to define the semantics of a programming language and casts the synthesis problem as a proof search \cite{SemGuSToolkitJohnsonRRD24,semgus}.
This line of work has an objective different from ours, but is highly compatible with our solver.
Having the semantics of a programming language available within \solvername would allow the user to specify more effective constraints over the semantics of programs.

\paragraph{Counter-example guided synthesis.}
This 
family of synthesisers \cite{cegis} uses constraint solvers to verify whether a proposed program is a solution to the provided specification.
If it is not, then the verifier provides a counter-example for the synthesiser.
In contrast, our approach uses constraints to aid search by pruning the program space on the program level, not by introducing new examples.

\paragraph{Automated program repair (APR)}
APR \cite{APRIntroGouesPR19} can be formulated as a program synthesis problem starting from partial programs.
Recent approaches like CPR \cite{CPRShariffdeenNGR21} and ExtractFix \cite{ExtractFixGaoWDJXR21} extract and propagate constraints to remove concrete subprograms from their search space like Popper. 
EffFix \cite{EffFixZhangCSMR25} uses another notion of constraints, using probabilities to indirectly describe constraints and guide the search.
Having a similar goal of removing unwanted programs as early as possible, we posit that our constraint formalism could benefit the APR community.
Further, APR applications could help to extend our limited library of constraints.

\section{Conclusion and Future Work}
\label{sec:conclusion_future}

This work takes a new approach to taming the size of the search space in program synthesis.  Specifically, we focused on defining the syntactic space of programs, beyond formulating a context-free grammar. 
Such syntactic constraints open the possibility of introducing language biases into the program space, inspired by \citet{POPPER}. 

The contributions of this paper are made concrete in a novel constraint solver, \solvername, for propagating and solving syntactic constraints efficiently.  
\solvername adopts a hybrid method of two sub-solvers: (1) the decomposition solver can be used to construct simple search spaces, represented by uniform trees, and (2) the uniform solver is restricted to one shape, but can leverage common CSP techniques.

We compared \solvername's constrained program enumeration against plain program enumeration.
Although the effectiveness highly depends on the grammar and its constraints, the results show that constraint propagation significantly outperforms a retrospective checking the constraints. 
Furthermore, we showed that combining grounded constraints into a single first-order constraint can further reduce the number of search nodes. 
This is an expected result as a combined constraint is able to make deductions based on the interaction of the grounded constraints and therefore has stronger inference.
Eventually, we showed that constraints can be used beyond removing symmetries, guiding state-of-the-art to find solutions quicker.

In future work, we aim to extend the library of constraints presented to other families of constraints.
Second, while for now we have only considered constraining the grammar itself, we did not take the problem to be solved into account.
A constraint extractor could take the input--output examples of the problem and generate grammar constraints for the specific problem instance. 
We also aim to generate new constraints based on failed input--output examples during search. 

\bibliographystyle{ACM-Reference-Format}
\bibliography{bibfile}

\appendix

\newpage
\section{Constraint Propagation Algorithms}
Concrete implementations of the respective propagation algorithms.
A solver state is set to infeasible if the current tree cannot contain a valid solution.

\subsection{Forbid}
\begin{algorithm}[H]
        \begin{algorithmic}
        \STATE node $\gets$ get\_node\_at\_location(solver, constraint.path)
        \STATE match $\gets$ pattern\_match(node, constraint.tree)
        \IF{match isa HardFail}
            \STATE deactivate!(solver, constraint)
        \ELSIF{match isa Success}
            \STATE set\_infeasible!(solver, constraint)
        \ELSIF{match isa SuccessWhenHoleAssignedTo}
            \STATE remove!(solver, match.hole, match.rule)
            \STATE deactivate!(solver, constraint)
        \ENDIF
    \end{algorithmic}
    \caption{Propagation of a local forbid constraint.}
    \label{alg:propagatelocalforbidden}
\end{algorithm}

\subsection{Contains}
Contains is closely related to the unique constraint. 
\begin{algorithm}[h!]
        \begin{algorithmic}
        \FOR{$i \in$ constraint.indices}
            \STATE candidate $\gets$ constraint.candidates[$i$]
            \STATE match $\gets$ pattern\_match(candidate, constraint.tree)
            \IF{match isa HardFail}
                \STATE remove!(constraint.indices, $i$)
            \ELSIF{match isa Success}
                \STATE deactivate!(solver, constraint)
            \ENDIF
        \ENDFOR
        \STATE $n \gets$ length(constraint.indices)
        \IF{$n$ == $0$}
            \STATE set\_infeasible!(solver, constraint)
        \ELSIF{$n$ == $1$}
            \STATE $i \gets$ minimum(constraint.indices)
            \STATE candidate $\gets$ constraint.candidates[$i$]
            \STATE result $\gets$ make\_equal!(solver, candidate, constraint.tree)
            \IF{result isa HardFail}
                \STATE set\_infeasible!(solver, constraint)
            \ELSIF{result isa Success}
                \STATE deactivate!(solver, constraint)
            \ENDIF
        \ENDIF
    \end{algorithmic}
    \caption{Propagation of a 'contains sub-tree' constraint.}
    \label{alg:propagatelocalcontainssubtree}
\end{algorithm}

\subsection{Ordered Constraint}
The \texttt{Ordered} constraint needs to be enforced wherever the shape can match.
A local ordered constraint is posted.

\begin{algorithm}[H]
        \begin{algorithmic}
        \STATE node $\gets$ get\_node\_at\_location(solver, constraint.path)
        \STATE vars $\gets$ Dict()
        \STATE match $\gets$ pattern\_match(node, constraint.tree, vars)
        \IF{match isa Fail}
            \STATE deactivate!(solver, constraint)
        \ELSIF{match isa Success}
            \STATE should\_deactivate $\gets$ true
            \STATE $n$ $\gets$ length(constraint.order)
            \FOR{name1, name2 $\in$ zip(constraint.order[1:$n$-1], constraint.order[2:$n$])}
                \STATE result $\gets$ make\_less\_than\_or\_equal!(solver, vars[name1], vars[name2])
                \IF{result isa HardFail}
                    \STATE set\_infeasible!(solver)
                \ELSIF{result isa SoftFail}
                    \STATE should\_deactivate $\gets$ false
                \ENDIF
            \ENDFOR
            \IF{should\_deactivate}
                \STATE deactivate!(solver, constraint)
            \ENDIF
        \ENDIF
    \end{algorithmic}
    \caption{Propagation of a local ordered constraint.}
    \label{alg:propagatelocalordered}
\end{algorithm}

\section{Search Algorithms using \solvername}
We describe how to implement two families of search algorithms in our paper. 
The concrete algorithms are shown in \Cref{alg:topdowniterator} and \Cref{alg:bottomupiterator}.

\begin{algorithm}[t]
    \begin{algorithmic}
        \STATE solver, pq $\gets$ init\_solver(grammar, starting\_symbol)
        \WHILE{len(pq) $> 0$}
            \STATE state $\gets$ pq.dequeue()
            \IF{state isa DecompSolverState}
                \STATE solver.load\_state(state)
                \STATE hole $\gets$ hole\_heuristic(solver.get\_tree())
                \IF{!hole}
                    \STATE uniform\_solver $\gets$ UniformSolver(grammar, solver.get\_tree())
                    \STATE pq.enqueue(uniform\_solver, priority\_function(...))
                \ELSE
                    \FOR{uniform\_domain $\in$ partition(hole.domain)}
                        \STATE solver.remove\_all\_but!(hole, uniform\_domain)
                        \STATE pq.enqueue(solver.get\_state(), priority\_function(...))
                    \ENDFOR
                \ENDIF
            \ELSIF{state isa UniformSolver}
                \STATE program $\gets$ next\_solution(state)
                \STATE pq.enqueue(state, priority\_function(...))
                \item \textbf{yield} program
            \ENDIF
        \ENDWHILE
    \end{algorithmic}
    \caption{Top-Down Iterator algorithm using \solvername's interface}
    \label{alg:topdowniterator}
\end{algorithm}

\begin{algorithm}[t]
    \begin{algorithmic}
        \STATE solver $\gets$ init\_solver(grammar)
        \STATE bank $\gets$ uniform\_tree(terminals(grammar))
        \STATE uniform\_partitions $\gets$ partition(grammar.rules)
       \FOR{ $n\in$ 1:max\_depth}
            \STATE uniform\_trees $\gets$ substitute\_subtree(uniform\_partitions, bank)
            \STATE pq.enqueue(UniformSolver(grammar, tree)) for tree in uniform\_trees
            \WHILE{len(pq) $> 0$}
                \STATE state $\gets$ pq.dequeue()
                \STATE program $\gets$ next\_solution(state)
                \STATE pq.enqueue(state, priority\_function(...))
                \STATE update\_bank(bank, program) \COMMENT{Adds program to a uniform tree}
                \STATE \textbf{yield} program
            \ENDWHILE
        \ENDFOR
    \end{algorithmic}
    \caption{Bottom-Up Iterator algorithm using \solvername's interface.}
    \label{alg:bottomupiterator}
\end{algorithm}

\section{Grammars}

In this section, we describe all grammars and constraints used in the experiments. 
We use code as a more compact representation of constraints. 
First, note that a \texttt{RuleNode} denotes a value node and is used to describe all ASTs.
Second, we use a special variant of the \texttt{Forbid} constraint for the Robots domain:
\texttt{ForbiddenSequence}.
For example, the constraint \texttt{ForbiddenSequence([+,:a, +])} is violated if on any path from the root to any leaf the $+$ occurs twice.
All proofs and definitions remain valid for this constraint.

\subsection{Arithmetic/LIA Grammar}
The arithmetic grammar is closely related to the LIA grammar (\Cref{fig:LIA_example}) can be used to synthesize simple arithmetic expressions with an input symbol $x$. For example, given the IO examples $1 \rightarrow 2$ and $5 \rightarrow 6$, the synthesizer will return $x + 1$ as the target program. 

This variant of the grammar includes 10 terminal rules for constants, instead of 1.

\begin{figure}[h!]
    \centering
    \begin{lstlisting}[language=Julia]
grammar = @csgrammar begin
    Int = Int + Int
    Int = Int * Int
    Int = Int - Int
    Int = |(0:9)
    Int = x
end
    \end{lstlisting}
    \caption{An arithmetic grammar, closely related to LIA.}
    \label{fig:simplifiedarithmeticgrammar}
\end{figure}

\noindent The constraints for the arithmetic grammar (See \Cref{fig:simplifiedarithmeticgrammarconstraints}) all break semantic symmetries in integer arithmetic. They forbid trivial cases like adding 0 and multiplying by 1. The ordered constraints will be used to break the commutativity of $+$ and $\times$. This is not an exhaustive list of symmetry-breaking constraints, and could be extended to break more symmetries.

\begin{figure}[h!]
    \centering
    \begin{lstlisting}[language=Julia]
Forbidden(RuleNode(times, [VarNode(:a), RuleNode(zero)]))
Forbidden(RuleNode(minus, [VarNode(:a), VarNode(:a)]))
Forbidden(RuleNode(minus, [VarNode(:a), RuleNode(zero)]))
Forbidden(RuleNode(plus, [VarNode(:a), RuleNode(zero)]))
Forbidden(RuleNode(times, [VarNode(:a), RuleNode(one)]))
Forbidden(RuleNode(minus, [
    RuleNode(times, [VarNode(:a), RuleNode(two)])
    VarNode(:a)
]))
Forbidden(RuleNode(plus, [VarNode(:a), VarNode(:a)]))
Forbidden(RuleNode(minus, [
    RuleNode(times, [VarNode(:a), RuleNode(three)])
    VarNode(:a)
]))
Forbidden(RuleNode(plus, [RuleNode(zero), VarNode(:a)]))
Forbidden(RuleNode(times, [RuleNode(zero), VarNode(:a)]))
Forbidden(RuleNode(times, [RuleNode(one), VarNode(:a)]))
Ordered(RuleNode(plus, [VarNode(:a), VarNode(:b)]), [:a, :b])
Ordered(RuleNode(times, [VarNode(:a), VarNode(:b)]), [:a, :b])\end{lstlisting}
    \caption{13 constraints on the arithmetic grammar defined in Figure \ref{fig:simplifiedarithmeticgrammar}}
    \label{fig:simplifiedarithmeticgrammarconstraints}
\end{figure}

\subsection{Robots Grammar}
The robot grammar \cite{BruteCropperD20} describes programs in the Robot Environment. 
This environment consists of an $n \times n$ grid, a robot and a ball. 
A state describes the position of the robot and the ball. 
The task is to learn to transform the initial state to the final state by moving the robot and letting it grab and drop the ball.

We will simplify the robot grammar by omitting conditional statements, as the environment is too simple for them to be needed. 
\noindent In the upcoming experiments, we will use the robot grammar as defined in Figure \ref{fig:simplifiedrobotgrammar}, equipped with the grammar constraints in Figure \ref{fig:simplifiedrobotgrammarconstraints}. 
\begin{figure}[h!]
    \centering
    \begin{lstlisting}[language=Julia]
grammar = @csgrammar begin
    Sequence = (moveRight(); Sequence)
    Sequence = (moveDown(); Sequence)
    Sequence = (moveLeft(); Sequence)
    Sequence = (moveUp(); Sequence)
    Sequence = (drop(); Sequence)
    Sequence = (grab(); Sequence)
    Sequence = return
end
    \end{lstlisting}
    \caption{A simplified grammar for the robot environment.}
    \label{fig:simplifiedrobotgrammar}
\end{figure}

\begin{figure}
    \centering
    \begin{lstlisting}[language=Julia]
# the robot can drop and grab at most once
Unique(r_drop)
Unique(r_grab)

# shortest path constraints
ForbiddenSequence([r_left, r_right], ignore_if = [r_drop, r_grab])
ForbiddenSequence([r_right, r_left], ignore_if = [r_drop, r_grab])
ForbiddenSequence([r_up, r_down], ignore_if = [r_drop, r_grab])
ForbiddenSequence([r_down, r_up], ignore_if = [r_drop, r_grab])

# symmetry breaking constraints
ForbiddenSequence([r_down, r_right], ignore_if = [r_drop, r_grab])
ForbiddenSequence([r_down, r_left], ignore_if = [r_drop, r_grab])
ForbiddenSequence([r_up, r_right], ignore_if = [r_drop, r_grab])
ForbiddenSequence([r_up, r_left], ignore_if = [r_drop, r_grab])
    \end{lstlisting}
    \caption{10 constraints on the robot grammar defined in Figure \ref{fig:simplifiedrobotgrammar}}
    \label{fig:simplifiedrobotgrammarconstraints}
\end{figure}

\subsection{Symbolic Grammar}
\label{sec:symbolicgrammar}
The symbolic grammar (see Figure \ref{fig:symbolicgrammar}), and its 21 constraints have no semantics and exists solely to push the constraint solver to its limits.

Semantic grammars can be constrained to prune a large amount of the program space, but a significant amount of it remains valid. In traditional CP, often only a few, if any at all, solutions exist. In such cases, most time is spent in finding even a single solution. To mimic this setting and push constraint propagation to its limit, we will consider the highly constrained symbolic grammar, in which only a handful of programs satisfy the constraints.

\begin{figure}[h!]
    \centering
    \begin{lstlisting}[language=Julia]
grammar = @csgrammar begin
    S = t1 #terminals               # rule 1
    S = t2                          # rule 2
    S = t3                          # rule 3
    S = t4                          # rule 4
    S = u1(S) #unary functions      # rule 5
    S = u2(S)                       # rule 6
    S = u3(S)                       # rule 7
    S = b1(S, S) #binary functions  # rule 8
    S = b2(S, S)                    # rule 9
    S = b3(S, S)                    # rule 10
end\end{lstlisting}
    \caption{A symbolic grammar without any semantics.}
    \label{fig:symbolicgrammar}
\end{figure}

\subsection{Lists Grammar}
The Lists grammar (see \Cref{fig:listgrammar}) is a semantic grammar that can be used for list manipulations. It supports basic list operations and can be used to construct a program that takes two input integers $x$ and $y$ and returns a list.

To illustrate program synthesis using this grammar, consider the following two IO examples: 

\[
\begin{aligned}
(x = 0,\; y = 1) &\;\longrightarrow\; [1, 3],\\
(x = 5,\; y = 4) &\;\longrightarrow\; [3, 9].
\end{aligned}
\]

The intended behaviour of the program is to return a sorted list of the sum of the input values and the constant $3$:

\begin{equation*}
    \verb|sort!(push!(push!([], sum(push!(push!([], y), x))), 3))|
\end{equation*}

The constraints for the list grammar (See Figure \ref{fig:listgrammarconstraints}) eliminate semantically redundant programs. For example, by forbidding sorting a list twice in a row.
We also define a domain rule node representing all unary functions. 
Two of the constraints use this node to forbid unary functions on an empty list (line 14) and a singleton list (line 15), respectively. 
The constraint on line 17 aims to forbid reversing a constant list. 
For example, \verb|reverse!(push!(push!([],2),x))))| can be forbidden, as it is already represented by \verb|push!(push!([],x),2)|.

Just like for the arithmetic grammar, the provided list of constraints is not exhaustive and could be further extended. 
However, they do prune enough of the program space to evaluate the performance of constraint propagation in the upcoming sections.

\begin{figure}[h!]
    \centering
    \begin{lstlisting}[language=Julia]
grammar = @csgrammar begin
    List = []
    List = push!(List, Int)
    List = reverse!(List)
    List = sort!(List)
    List = append!(List, List)
    Int = maximum(List)
    Int = minimum(List)
    Int = sum(List)
    Int = prod(List)
    Int = |(1:3)
    Int = x
    Int = y
end
    \end{lstlisting}
    \caption{A grammar with basic list operations.}
    \label{fig:listgrammar}
\end{figure}

\begin{figure}[h!]
    \centering
    \begin{lstlisting}[language=Julia]
A = VarNode(:a)
B = VarNode(:b)
V = VarNode(:v)
unaryfunction(node::AbstractRuleNode) = DomainRuleNode(grammar, 
    [_reverse, _sort, _max, _min, _sum, _prod], [node])

Forbidden(reverse(reverse(A)))
Forbidden(sort(reverse(A)))
Forbidden(sort(sort(A)))
Forbidden(sort(append(A, reverse(B))))
Forbidden(sort(append(A, sort(B))))
Forbidden(sort(push(sort(A), V)))

Forbidden(unaryfunction(empty))
Forbidden(unaryfunction(push(empty, V)))

ForbiddenSequence([_reverse, _empty], ignore_if=[_sort, _append])
ForbiddenSequence([_append, _empty], ignore_if=[_reverse, _sort])
    \end{lstlisting}
    \caption{10 constraints on the list grammar defined in Figure \ref{fig:listgrammar}}
    \label{fig:listgrammarconstraints}
\end{figure}

\subsection{SLIA and BV}
SyGuS SLIA and BV are two program synthesis benchmarks from the SyGuS Challenge 2019\cite{SyGuS2019}, commonly used to benchmark synthesizers.
Exploring constraints in real-world program synthesis benchmarks, we use three sets of constraints.
\begin{enumerate}
    \item \textbf{No constraints:} --
    \item \textbf{Constraints removing redundancies:} Programs must contain the input, breaking symmetry of equality and other commutative operators. For BV these are bitvector operations including addition, xor, and, or, nand, and nor.
    \item \textbf{Removing redundancies + \textit{useful} constraints}: For BV, adding to the previous constraints, we add the constraint \texttt{ForbiddenSequence([i,i]} for all possible grammar indices of non-terminal symbols. 
    For SLIA, we add \texttt{Contains(i)} for every rule index i that is terminal, not an input, not an empty string or space, and that has the String type.
\end{enumerate}

\section{Ablation Studies and Experiments}

\subsection{Program Space Reduction}
We compare the runtime of enumeration without constraints vs our propagation system.
The results are shown in \Cref{fig:AblationStudyRuntimeChecking}.
We see that our propagation is strictly faster than pure enumeration, removing many programs before they are enumerated.

\begin{figure}[h!]
    \centering
    \includegraphics[width=0.55\textwidth]{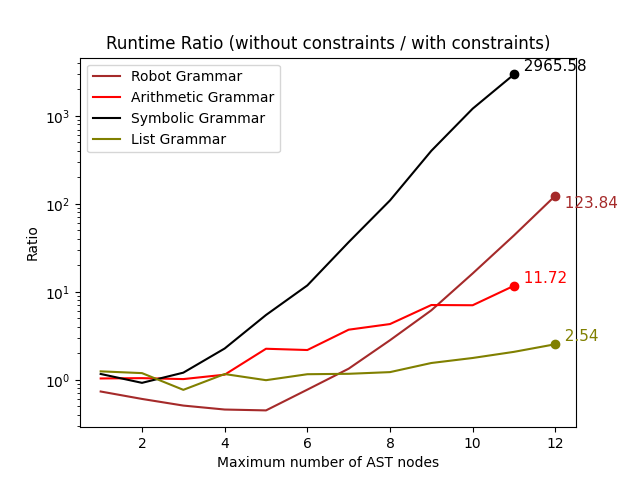}
    \caption{Runtime ratio of pure enumeration without constraints to our propagation system on the four grammars.}
    \label{fig:AblationStudyRuntimeChecking}
\end{figure}

We further list the concrete reduction of the program spaces:

\subsubsection{Arithmetic Grammar}
See \Cref{tab:space_size_arithmetic}.

\begin{table}[h!]
\centering
\begin{tabular}{c|r|r||r|r}
\textbf{Size} & \multicolumn{2}{c||}{\textbf{Without Constraints}} & \multicolumn{2}{c}{\textbf{With Constraints}} \\
\cline{2-5}
 & \textbf{Program Space} & \textbf{Runtime (s)} & \textbf{Program Space} & \textbf{Runtime (s)} \\
\hline
1  & 11             & 0.000   & 11             & 0.000   \\
2  & 11             & 0.000   & 11             & 0.000   \\
3  & 374            & 0.000   & 201            & 0.001   \\
4  & 374            & 0.000   & 201            & 0.001   \\
5  & 24 332         & 0.039   & 7 798          & 0.018   \\
6  & 24 332         & 0.039   & 7 798          & 0.018   \\
7  & 2 000 867      & 3.190   & 383 688        & 0.814   \\
8  & 2 000 867      & 3.852   & 383 688        & 0.819   \\
9  & 184 632 701    & 361.538 & 21 192 628     & 43.842  \\
10 & 184 632 701    & 357.280 & 21 192 628     & 45.778  \\
11 & 18 265 184 267 & 35048.885 & 1 254 647 849 & 2592.975 \\
\end{tabular}
\caption{Comparison of program space and runtime with and without constraints in the Arithmetic Grammar}
\label{tab:space_size_arithmetic}
\end{table}

\subsubsection{Robots Grammar}
See \Cref{tab:space_size_robots}.
\begin{table}[h!]
\centering
\begin{tabular}{c|r|r||r|r}
\textbf{Size} & \multicolumn{2}{c||}{\textbf{Without Constraints}} & \multicolumn{2}{c}{\textbf{With Constraints}} \\
\cline{2-5}
 & \textbf{Program Space} & \textbf{Runtime (s)} & \textbf{Program Space} & \textbf{Runtime (s)} \\
\hline
1  & 1                 & 0.000   & 1             & 0.000   \\
2  & 7                 & 0.000   & 7             & 0.000   \\
3  & 43                & 0.000   & 43            & 0.000   \\
4  & 259               & 0.000   & 133           & 0.002   \\
5  & 1 555             & 0.002   & 407           & 0.004   \\
6  & 9 331             & 0.017   & 1 457         & 0.022   \\
7  & 55 987            & 0.093   & 4 025         & 0.071   \\
8  & 335 923           & 0.492   & 11 001        & 0.209   \\
9  & 2 015 539         & 3.270   & 29 666        & 0.538   \\
10 & 12 093 923        & 23.728  & 67 129        & 1.306   \\
11 & 72 559 411        & 141.965 & 161 965       & 2.855   \\
12 & 435 356 467       & 844.853 & 384 853       & 6.794   \\
13 & 2 612 138 803     & (*)     & 292 741       & 12.288  \\
14 & 15 672 832 819    & (*)     & 488 257       & 23.051  \\
15 & 94 037 996 913    & (*)     & 805 863       & 42.535  \\
16 & 564 221 981 491   & (*)     & 1 368 153     & 69.516  \\
17 & 3 385 341 988 947 & (*)     & 1 839 991     & 96.197  \\
18 & 20 312 051 933 683& (*)     & 2 708 965     & 182.691 \\
19 & 121 871 948 002 803 & (*)   & 3 907 683     & 296.374 \\
20 & 731 231 688 012 955 & (*)   & 5 495 589     & 447.164 \\
21 & 4 387 390 128 075 571 & (*) & 7 610 421     & 681.843 \\
\end{tabular}
\caption{Program space and runtime comparison with and without constraints for size 1 to 21 in the Robots Grammar. 
(*) Instead of actual enumeration, the theoretical number of programs was calculated with  $\sum_{k=0}^{n-1} 6^k$.}
\label{tab:space_size_robots}
\end{table}

\subsubsection{Symbolic Grammar}
See \Cref{tab:space_size_symbolic}.

\begin{table}[h!]
\centering
\begin{tabular}{c|r|r||r|r}
\textbf{Size} & \multicolumn{2}{c||}{\textbf{Without Constraints}} & \multicolumn{2}{c}{\textbf{With Constraints}} \\
\cline{2-5}
 & \textbf{Program Space} & \textbf{Runtime (s)} & \textbf{Program Space} & \textbf{Runtime (s)} \\
\hline
1  & 4              & 0.000     & 0        & 0.000   \\
2  & 16             & 0.000     & 0        & 0.000   \\
3  & 100            & 0.000     & 0        & 0.001   \\
4  & 640            & 0.001     & 0        & 0.002   \\
5  & 4 708          & 0.010     & 0        & 0.007   \\
6  & 35 920         & 0.072     & 0        & 0.020   \\
7  & 287 236        & 0.550     & 0        & 0.056   \\
8  & 2 355 328      & 4.469     & 0        & 0.146   \\
9  & 19 763 524     & 43.766    & 0        & 0.406   \\
10 & 168 628 240    & 377.454   & 11       & 1.143   \\
11 & 1 459 357 732  & 3213.797  & 3        & 3.011   \\
12 & -              & -         & 11       & 8.752   \\
13 & -              & -         & 108      & 24.491  \\
14 & -              & -         & 2597     & 71.632  \\
15 & -              & -         & 27 667   & 210.122 \\
16 & -              & -         & 345 428  & 653.034 \\
\end{tabular}
\caption{Comparison of program space and runtime with and without constraints in the Symbolic Grammar}
\label{tab:space_size_symbolic}
\end{table}

\subsubsection{Lists Grammar}
See \Cref{tab:space_size_lists}.
\begin{table}[h!]
\centering
\begin{tabular}{c|r|r||r|r}
\textbf{Size} & \multicolumn{2}{c||}{\textbf{Without Constraints}} & \multicolumn{2}{c}{\textbf{With Constraints}} \\
\cline{2-5}
 & \textbf{Program Space} & \textbf{Runtime (s)} & \textbf{Program Space} & \textbf{Runtime (s)} \\
\hline
1  & 1             & 0.000   & 1             & 0.000   \\
2  & 3             & 0.000   & 1             & 0.000   \\
3  & 13            & 0.000   & 6             & 0.001   \\
4  & 51            & 0.001   & 20            & 0.001   \\
5  & 217           & 0.002   & 31            & 0.003   \\
6  & 951           & 0.010   & 56            & 0.008   \\
7  & 4 297         & 0.036   & 206           & 0.018   \\
8  & 19 887        & 0.124   & 1 026         & 0.042   \\
9  & 93 757        & 0.513   & 1 656         & 0.092   \\
10 & 448 875       & 2.172   & 5 381         & 0.216   \\
11 & 2 176 261     & 8.963   & 16 006        & 0.580   \\
12 & 10 663 563    & 37.675  & 53 631        & 1.159   \\
13 & 52 724 209    & 150.302 & 166 881       & 3.010   \\
14 & 262 718 895   & 771.796 & 548 256       & 6.346   \\
15 & 1 317 979 105 & 3629.086 & 1 779 631     & 15.498  \\
16 & -             & -       & 5 817 631     & 36.277  \\
17 & -             & -       & 19 329 756    & 89.173  \\
18 & -             & -       & 63 692 256    & 250.183 \\
19 & -             & -       & 213 391 631   & 652.133 \\
\end{tabular}
\caption{Comparison of program space and runtime with and without constraints in the  Lists Grammar}
\label{tab:space_size_lists}
\end{table}

\subsection{Performance of the Uniform Solver}
So far, all experiments have been executed using iterators that use both the decomposition and uniform solvers. 
The decomposition solver is used to propagate constraints on non-uniform trees, which has the potential to eliminate an entire search branch of uniform trees. 
It also uses memory-intensive state management to ensure programs are enumerated in increasing order of size. The uniform solver has efficient state management but is restricted to solving uniform trees and a DFS.
\begin{figure}[h!]
    \centering
    \includegraphics[width=0.75\textwidth]{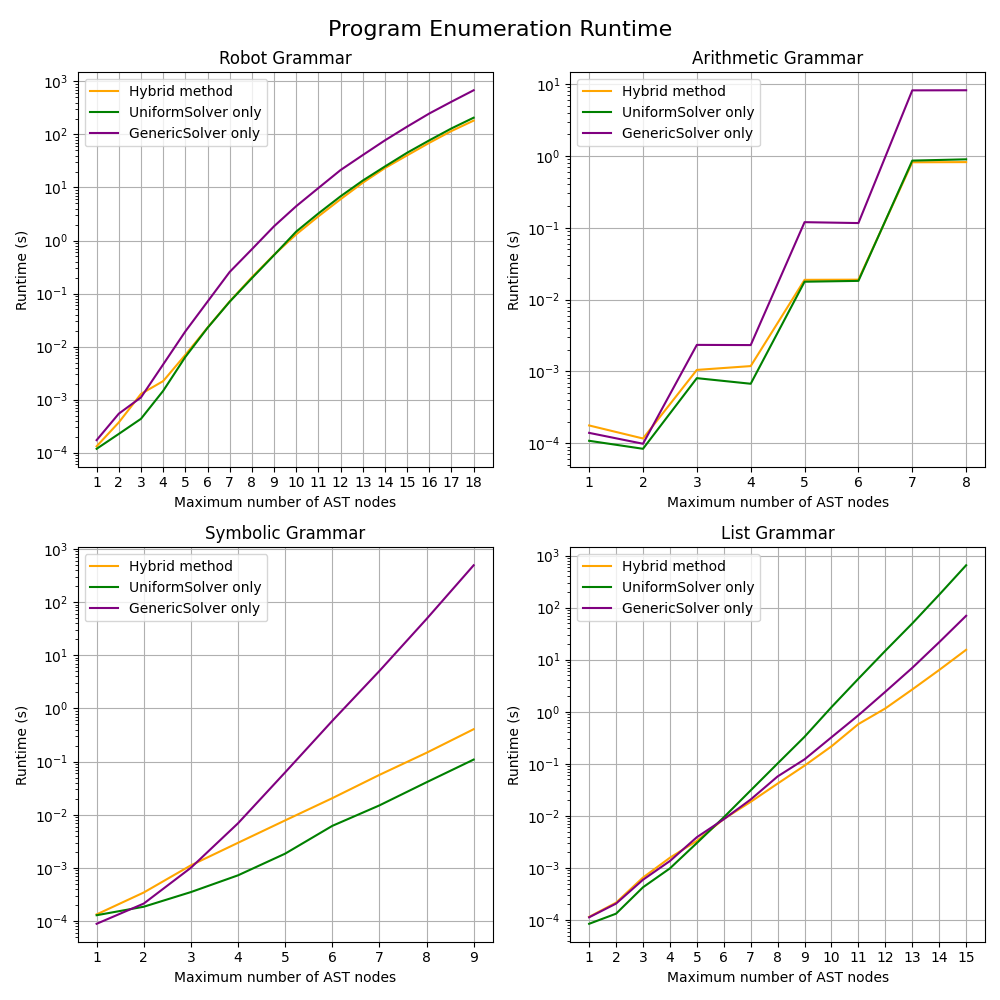}
    \caption{An ablation study. Comparing the runtime of the overall search procedure using both or only 1 of the 2 built-in solvers.}
    \label{fig:AblationStudyRuntime}
\end{figure}

In this section, we conduct an ablation study. 
Instead of just using the proposed hybrid method, we measure the performance by enumerating the program space using only one of the implemented solvers. 

More precisely, we will compare the following three methods:
\begin{enumerate}
    \item Hybrid method: The decomposition solver is used to enumerate uniform trees, and uniform solvers are used to enumerate all complete programs of each uniform tree, as described in \Cref{alg:topdowniterator}.
    \item Uniform solver only: We ignore constraints in the decomposition solver, and only start propagating constraints once a uniform tree is reached.
    \item Decomposition solver only: We never dispatch to the uniform solver. Even uniform trees will be expanded by the decomposition solver
\end{enumerate}

First, using only the decomposition solver does not rely on simple program spaces. 
All holes are treated as non-uniform holes, which significantly reduces the inference power of the constraints.
Second, only using the uniform solver, we can use constraints once they are broken down into local constraints on a uniform hole.
In this setting, all constraints are always posted on all feasible uniform holes. 

\Cref{fig:AblationStudyRuntime} holds the result of the ablation study. 
We see that the hybrid method always outperforms using only the decomposition solver. 
In both the hybrid and decomposition solver methods, the number of uniform trees is exactly the same. 
In other words, the results indicate that solving a uniform tree with the uniform solver, optimized for solving uniform trees, outperforms the decomposition solver.

For the symbolic grammar, only using the uniform solver outperforms the hybrid method (see \Cref{fig:AblationStudyRuntime}). 
This indicates that the overhead of propagating constraints in the decomposition solver does not outweigh the gained inference. 
This suspicion is confirmed by the follow-up experiment in \Cref{fig:AblationStudyUniformTrees1}. 
We see that the decomposition solver (orange line) only eliminates a negligible amount of uniform trees, which means the inference is weak. 
For this particular grammar, using the decomposition solver is not worth its overhead.

On the contrary, for the list grammar (see Figure \ref{fig:AblationStudyRuntime}), the hybrid method outperforms only using the uniform solver. 
This indicates that the decomposition solver has strong inference for this grammar. 
Again, this suspicion is confirmed by the follow-up experiment in Figure \ref{fig:AblationStudyUniformTrees2}. 
We see that the decomposition solver (orange line) eliminates a significant portion of uniform trees. 
A reduced number of uniform trees means that fewer uniform solvers have to be instantiated. 
Hence, the hybrid method will outperform only using the uniform solver.

\begin{figure}[h!]
  \centering
  \begin{subfigure}{0.38\textwidth}
    \centering
    \includegraphics[width=\textwidth]{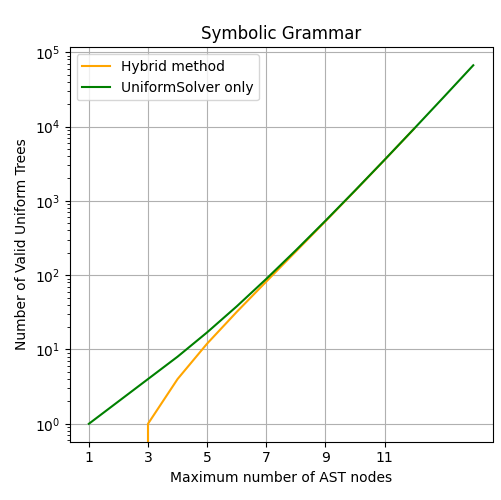}
    \caption{Symbolic Grammar}
    \label{fig:AblationStudyUniformTrees1}
  \end{subfigure}
  \hfill
  \begin{subfigure}{0.38\textwidth}
    \centering
    \includegraphics[width=\textwidth]{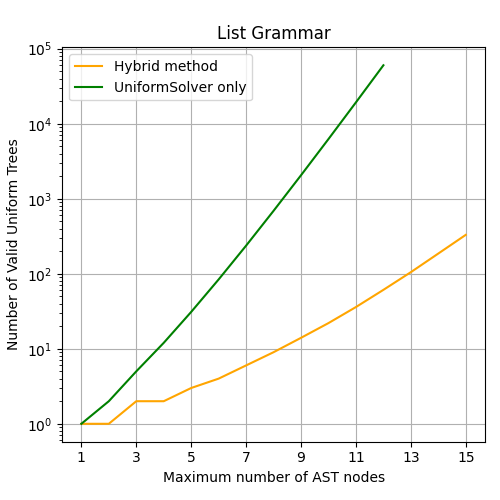}
    \caption{List Grammar}
    \label{fig:AblationStudyUniformTrees2}
  \end{subfigure}
  \caption{Comparing the total number of uniform trees with and without the decomposition solver.}
  \label{fig:AblationStudyUniformTrees}
\end{figure}

\subsection{First-order Constraints} 
We run two ablation studies on first-order constraints. 
First, we compare the first-order and grounded forbidden constraints.
We see a clear positive trend in the number of constraints and the number of propagate calls. 
This also correlates with the runtime of the propagation algorithm.

\begin{figure}[h!]
    \centering
    \includegraphics[width=0.75\textwidth]{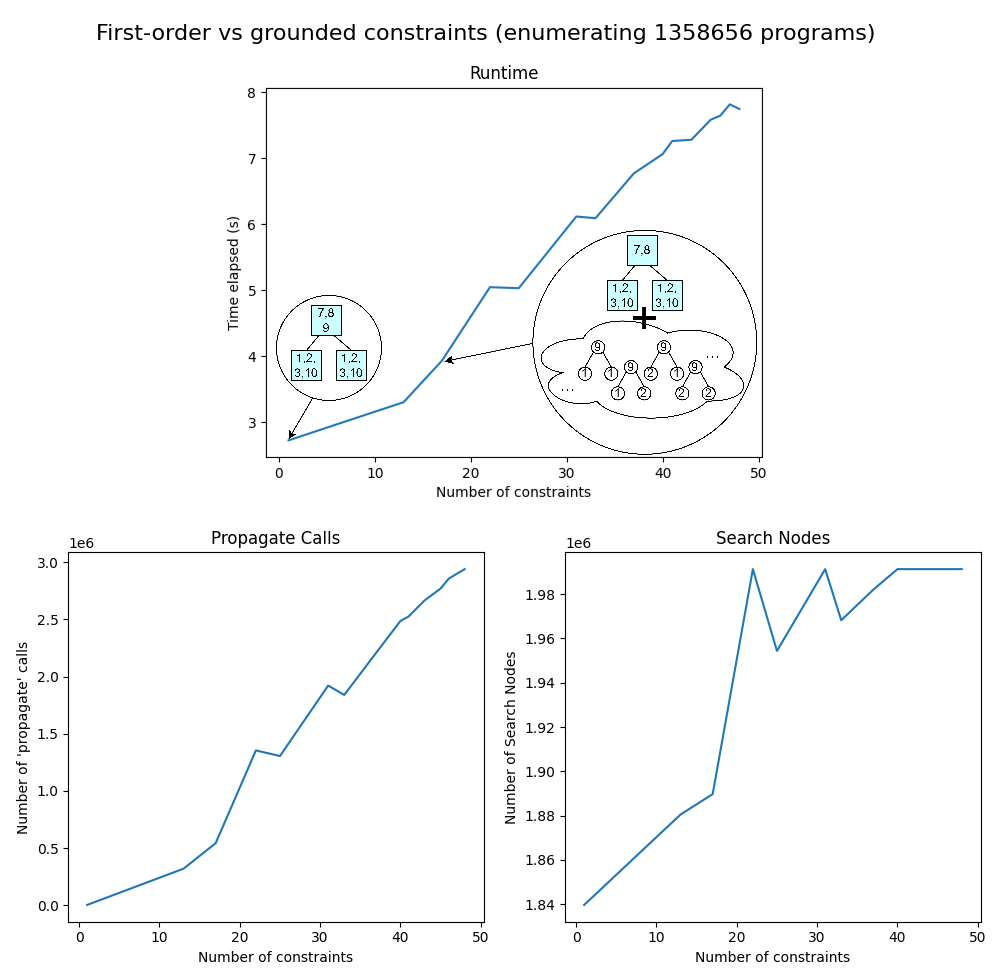}
    \caption{Enumerating programs of the Symbolic Grammar using different combinations of first-order and grounded forbidden constraints. The plots depict runtime, propagate calls, and search nodes, respectively.}
    \label{fig:FirstOrderConstraintsALL}
\end{figure}

\subsection{Bottlenecks}
\label{sec:bottlenecks}

\begin{figure}[]
    \centering
    \includegraphics[width=0.75\textwidth]{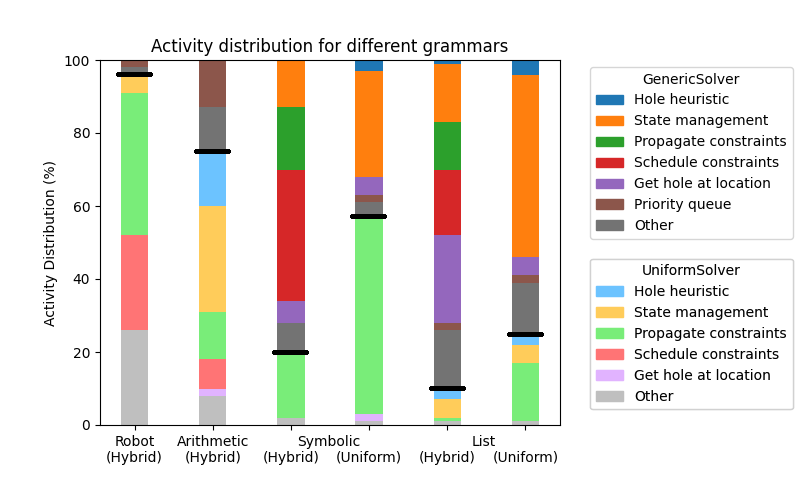}
    \caption{
    The activity distribution of the decomposition and uniform solver in the top-down iterator for several grammars. 
    For the Symbolic- and List grammars, an additional experiment without constraint propagation in the decomposition solver is included (labeled 'Uniform'). 
    The activity was measured using a profiler tool and grouped into high-level categories.}
    \label{fig:ActivityDistribution}
\end{figure}

In this section, we will look at the bottlenecks in the proposed methods to aid developers in optimizing the solvers. 
We can make several observations from the results in Figure \ref{fig:ActivityDistribution}.

For the Robot and Arithmetic grammars, we see that a significant amount of time is spent on posting constraints, almost as much as on propagating them.
This is because propagators are scheduled unnecessarily often. For the robot grammar, only $16\%$ of them were able to make any kind of deduction\footnote{A propagator makes a \textit{deduction} iff it prunes 1 or more rules from 1 or more domains.}.

For the Symbolic grammar, the issue is even more pronounced. 
We see that more time is spent posting than propagating. 
This indicates that many constraints are scheduled, but are unable to make any deductions. 
This is an expected result, as the decomposition solver deals with non-uniform trees, and most constraints can only make deductions if the tree's structure is known. 
In the ablation study (See \Cref{fig:AblationStudyRuntime}), we have seen that if we refrain from propagating constraints in the decomposition solver, we can improve the total runtime by a rough factor of 3. 
However, a more sustainable approach would be to limit the conditions for constraint propagation, thereby reducing the number of unnecessary ones.

\end{document}